%% file: ztp-tensor-decomposition.tex
\documentclass[numbers,webpdf,imaiai]{ima-authoring-template}

\graphicspath{{./images/}}
\theoremstyle{thmstyletwo}%
\newtheorem{theorem}{Theorem}
\newtheorem{lemma}[theorem]{Lemma}%
\usepackage{stmaryrd}

\usepackage{amssymb}
\usepackage{mathtools}
\usepackage[mathscr]{euscript}

\DeclareMathOperator*{\polylog}{poly\,log}
\DeclareMathOperator*{\argmax}{arg\,max}

\numberwithin{equation}{section}

\begin{document}

\DOI{DOI HERE}
\copyrightyear{2022}
\vol{00}
\pubyear{2022}
\access{Advance Access Publication Date: Day Month Year}
\appnotes{Paper}
\copyrightstatement{Published by Oxford University Press on behalf of the Institute of Mathematics and its Applications. All rights reserved.}
\firstpage{1}


\title[Zero-Truncated Poisson Regression]{Zero-Truncated Poisson Regression for Sparse Multiway\\ Count Data Corrupted by False Zeros}


\author{Oscar F. L\'{o}pez\ORCID{0000-0000-0000-0000}
\address{\orgdiv{Harbor Branch Oceanographic Institute}, \orgname{Florida Atlantic University}, \orgaddress{\street{5600 US 1 North}, \postcode{34946}, \state{Florida}, \country{U.S.A.}}}}
\author{Daniel M. Dunlavy\ORCID{0000-0000-0000-0000}
\address{\orgdiv{Machine Intelligence and Visualization}, \orgname{Sandia National Laboratories}, \orgaddress{\street{1515 Eubank SE}, \postcode{87123}, \state{New Mexico}, \country{U.S.A.}}}}
\author{Richard B. Lehoucq\ORCID{0000-0000-0000-0000}
\address{\orgdiv{Discrete Math and Optimization}, \orgname{Sandia National Laboratories}, \orgaddress{\street{1515 Eubank SE}, \postcode{87123}, \state{New Mexico}, \country{U.S.A.}}}}

\authormark{O. L\'{o}pez, D. Dunlavy, and R. Lehoucq}


\received{Date}{0}{Year}
\revised{Date}{0}{Year}
\accepted{Date}{0}{Year}


\abstract{
We propose a novel statistical inference methodology for multiway count data that is corrupted by false zeros that are indistinguishable from true zero counts. Our approach consists of zero-truncating the Poisson distribution to neglect all zero values. This simple truncated approach dispenses with the need to distinguish between true and false zero counts and reduces the amount of data to be processed. Inference is accomplished via tensor completion that imposes low-rank tensor structure on the Poisson parameter space. 
Our main result shows that an $N$-way rank-$R$ parametric tensor $\boldsymbol{\mathscr{M}}\in(0,\infty)^{I\times \cdots\times I}$ generating Poisson observations can be accurately estimated by zero-truncated Poisson regression from approximately $IR^2\log_2^2(I)$ non-zero counts under the nonnegative canonical polyadic decomposition. Our result also quantifies the error made by zero-truncating the Poisson distribution when the parameter is uniformly bounded from below. Therefore, under a low-rank multiparameter model, we propose an implementable approach guaranteed to achieve accurate regression in under-determined scenarios with substantial corruption by false zeros. Several numerical experiments are presented to explore the theoretical results.
}

\maketitle


\section{Introduction}
\label{sec:intro}
\input{intro}

\section{Numerical Experiments}
\label{sec:experiments}
\input{experiments}

\section{Main Theorems and Proofs}
\label{sec:proofs}
\input{main}

\section{Conclusions}
\label{sec:conclusions}
\input{conc}

\appendix

\section{Proofs of Lemmas}
\label{sec:proofslemmas}
\input{mainlemmas}

\section{Implementation Details of Numerical Experiments}
\label{sec:impl}
\input{impldetails}

\section*{Acknowledgements}
\label{sec:ack}
The authors would like to thank Jon Berry for helpful discussions and suggestions. This paper describes objective technical results and analysis. Any subjective views or opinions that might be expressed in the paper do not necessarily represent the views of the U.S. Department of Energy or the United States Government. 

\section*{Funding}
\label{sec:funding}
This work was supported by Sandia National Laboratories. Sandia National Laboratories is a multimission laboratory managed and operated by National Technology \& Engineering Solutions of Sandia, LLC, a wholly owned subsidiary of Honeywell International Inc., for the U.S. Department of Energy’s National Nuclear Security Administration under contract DE-NA0003525. 

\bibliographystyle{plain}
\bibliography{ztp-tensor-decomposition}

\end{document}

%% file: intro.tex
Count data arises in many data science applications including topic modeling \cite{An20,BlNgJo03,Gr22,Ho99}, document clustering \cite{AgZh12,StKaKu00} and classification \cite{GaJaCh20,KoJaHe19}, poll analysis \cite{Klimek16469}, network communications \cite{attributedSN,pmlr-v97-chiquet19a}, single photon count imaging \cite{osti_1347904,naturelidar}, and ecology \cite{bmpc:19}. Statistical interpretation of count data typically involves estimating parametric distributions likely to generate the counts via regression and maximum likelihood estimation \cite{sjams10.11648,doi:10.1080/00223890802634175,10.2307/270996}. Though useful for analysis and decision making, in most practical settings the collected data are corrupted by false counts that mislead the inference procedure. In particular, such arrays are frequently congested by zeros, either false or true, in an indistinguishable manner \cite{doi:10.1063/1.4980994,doi:10.1186/s12863-017-0561-z,doi:10.5539/ijsp.v7n3p22,https://doi.org/10.1002/sim.3699}. In the context of this paper, a portion of the zeros are considered false counts (whose locations are unknown)---i.e., erroneous counts, structural zeros denoting unobserved array entries, etc. The source of such corruption is largely an artifact of the standard practice to initialize arrays with all zero entries prior to data collection paired with flawed counting procedures. However, many probability distributions that govern the observed counts are expected to generate a large amount of true zero counts, e.g., Poisson and Bernoulli  distributions. This gives the set of zero values a central role in count data, where distinguishing and appropriately handling zero-congestion is crucial for accurate analysis and has long been a challenge in the field; see e.g. \cite{bmpc:19} for a discussion and many citations to this problem in the literature. We note that our setting differs from work in the context of \textit{overdispersion} \cite{sjams10.11648,doi:10.1080/00223890802634175} and \textit{zero-inflation} \cite{doi:10.1063/1.4980994,doi:10.1186/s12863-017-0561-z}, where the excessive zeros are considered as trustworthy data.

Further complicating the task of count data analysis is the inexorable growth in the volume and dimension of collected data---e.g., due to the expansion of global communication and social networks generating immense amounts of data to be mined. In such large-dimensional settings, multiway data analysis and \textit{tensor decompositions} (or factorizations) extract insight to interpret the role of each independent data component \cite{4538221}. When applied to tensors containing redundant and/or correlated information, such factorized representations additionally provide a compressive manner by which to process data that are otherwise too large to handle efficiently. Due to the relative simplicity of many data generation processes, the underlying multiway distributions can be modeled accurately by parametric tensors with few components relative to the ambient dimensions (i.e., low-rank tensors \cite{kolda}). For this reason, tensor decompositions are a numerically efficient tool to achieve multivariate statistical inference.

In this paper, we propose a novel statistical inference technique for multi-way count data that is saturated by false zero values. We truncate the multi-parameter \textit{Poisson model} to the positive integers, ignore zero values and treat the respective array entries as unobserved. Under a low-rank parametric tensor model, we achieve parameter estimation via \textit{tensor completion} that imposes large \textit{zero-truncated Poisson} likelihood using only the positive counts. In this manner, we exploit the low-dimensional structure found in many parametric models to accurately infer the underlying mean values of the entire volume in an under-determined setting that avoids false zero counts in regression. 

Our approach does not introduce additional parameters to be determined as described in the papers  \cite{doi:10.1063/1.4980994,https://doi.org/10.1002/sim.3699,doi:10.5539/ijsp.v7n3p22} and does not require the zeros to be classified as true or false counts as described in \cite{bmpc:19}. 
Furthermore, our setting is distinct from standard tensor and matrix completion problems \cite{5452187,5466511,PoissonMC,nuctensor,Gandy_2011,NIPS2014_c15da1f2} where the locations of unobserved entries is known \emph{a priori}. The difference may seem subtle, but our context is more complicated and common for count data since missing information exhibits itself as zeros that are indistinguishable from true null events of the data collection process. Our contribution is a simple and accurate approach that deals with zero-congestion in an efficient manner while reducing the potential for tuning and declassification errors. 

We begin with a theorem that elaborates our approach and its effectiveness to deal with false zero counts. The theorem summarizes our two main results (see Section~\ref{sec:proofs}), providing an error bound for parametric estimators with relatively large log-likelihood as a function of the data's factor dimensions along with the number of non-corrupt observations. The result compares our proposed method to the ideal estimator in which an ``oracle'' identifies the false zeros and Poisson regression can be applied on the true counts. Our main result states that our zero-truncated approach performs nearly as well as the oracle while remaining oblivious to the locations of false zeros when the Poisson parameter is uniformly bounded from below by zero. These implications are validated in Section~\ref{sec:experiments}, where numerical experiments present several realistic situations in which the performance of our zero-truncated paradigm is comparable to the oracle.

We first provide notation, definitions and a clear statement of the inference problem before we state the theorem. We use the conventions in \cite{Chi2012} and also rely on the notation of \cite{navid1,navid2,navid3,PoissonMC}. We focus on nonnegative tensors and their nonnegative \textit{canonical polyadic decomposition} (NNCP). 
Given $I_1,I_2,\cdots,I_N\in\mathbb{N}$ and a canonical polyadic tensor  $\boldsymbol{\mathscr{T}}\in\mathbb{R}_+^{I_1\times\cdots\times I_N}$ with nonnegative entries, we define the NNCP rank of $\boldsymbol{\mathscr{T}}$ as
\begin{align}  \label{t-nncp-rank}
\mbox{rank}_+(\boldsymbol{\mathscr{T}})\coloneqq \min\Bigg\{R\in\mathbb{N} \ \Big| \ \boldsymbol{\mathscr{T}}=\sum_{r=1}^{R} \textbf{a}_{r}^{(1)}\circ \textbf{a}_{r}^{(2)}\circ\cdots \circ \textbf{a}_{r}^{(N)} \ \ \mbox{with} \ \ \textbf{a}_{r}^{(n)}\in\mathbb{R}_{+}^{I_n} \ \ \forall r\in [R], n\in [N] \Bigg\},
\end{align}
where $[N]$ denotes the set $\{1,2,\cdots,N\}$ and  $\mathbb{R}_+$ denotes the values in $\mathbb{R}$ that are nonnegative. In other words, the NNCP rank is similar to the usual definition of CP rank \cite{kolda} but only applies to nonnegative tensors and imposes nonnegative constraints on the factors. Such nonnegative matrix and tensor decompositions have received increasing amounts of attention due to their uniqueness properties \cite{https://doi.org/10.1002/cem.1244}, resulting in an enhanced ability to extract meaningful data components sought by practitioners \cite{doi:10.1080/10556780801996244,Chi2012}.

%
The Poisson parameter tensor search space is
\begin{align} \label{t-search-space}
S_R^+(\beta,\alpha) \coloneqq \Big\{\boldsymbol{\mathscr{T}}\in\mathbb{R}^{I_1\times \cdots\times I_N} \ | \ \beta\leq t_{\textbf{i}}\leq \alpha \ \ \mbox{and} \ \ \mbox{rank}_+(\boldsymbol{\mathscr{T}})\leq R\Big\},
\end{align}
given a NNCP rank $R$
where $\textbf{i}=(i_1,i_2,\cdots,i_N)\in[I_1]\times[I_2]\times\cdots\times[I_N]$ denotes a multi-index, $t_{\textbf{i}}$ is the respective entry of $\boldsymbol{\mathscr{T}}$, and $0<\beta\leq\alpha$ are fixed but arbitrary bounds on the Poisson distribution parameters.

Our inference problem is to determine a low-rank Poisson parameter tensor $\boldsymbol{\mathscr{M}}\in S_R^+(\beta,\alpha)$ likely to generate observed count data $\boldsymbol{\mathscr{X}}\in\mathbb{Z}_+^{I_1\times \cdots\times I_N}$, where $\mathbb{Z}_+$ denotes nonnegative values in $\mathbb{Z}$. We assume the true Poisson events (or true counts) satisfy
\begin{equation}
	\label{poisson}
	x_{\textbf{i}}\sim \ \mbox{Poisson}(m_{\textbf{i}}), \ \ \textbf{i}\in\Omega
\end{equation}
for some subset $\Omega\subset[I_1]\times[I_2]\times\cdots\times[I_N]$. Outside of $\Omega$, the counts do not obey the Poisson generation model (\ref{poisson}) and consist of false zeros. 

The problem of estimating $\prod_kI_k$ parameters in $\boldsymbol{\mathscr{M}}$ from $|\Omega|<\prod_kI_k$ samples of count data in $\boldsymbol{\mathscr{X}}$ is under-determined. We circumvent this problem by imposing the low-rank assumption of the parameter model $\boldsymbol{\mathscr{M}}\in S_R^+(\beta,\alpha)$, which reduces the complexity of estimating the Poisson parameters to roughly determining $NR\sum_kI_k$ free variables (i.e., specifying the $\textbf{a}_r^{(n)}$'s in the NNCP decomposition (\ref{t-nncp-rank}) of $\boldsymbol{\mathscr{M}}$). By exploiting this low-dimensional structure, we now have a viable approach to solve our inference problem given a single instance of partially observed count data $\boldsymbol{\mathscr{X}}$.

In the ideal scenario that $\Omega$ can be identified, the low-rank factor model $\boldsymbol{\mathscr{M}}$ can be determined by optimizing the Poisson log-likelihood function on the true counts
\begin{equation}
	\label{likelihood}
	f_{\Omega}(\boldsymbol{\mathscr{M}},\boldsymbol{\mathscr{X}}) \coloneqq \sum_{\textbf{i}\in\Omega}x_{\textbf{i}}\log\left(m_{\textbf{i}}\right)-m_{\textbf{i}} - \log(x_{\textbf{i}}!).
\end{equation}
However, $\Omega$ is not known in general and estimators utilizing (\ref{likelihood}) will be known as \emph{oracle} estimators. 
Instead, we propose to compute a parameter model by optimizing the zero-truncated Poisson log-likelihood function
\begin{equation}
	\label{likelihood0}
	\tilde{f}_{\Gamma}(\boldsymbol{\mathscr{M}},\boldsymbol{\mathscr{X}}) \coloneqq \sum_{\textbf{i}\in\Gamma}x_{\textbf{i}}\log\left(m_{\textbf{i}}\right)-\log\left(\exp(m_{\textbf{i}}) - 1\right) - \log(x_{\textbf{i}}!)
\end{equation}
where $\Gamma$ provides the indices of non-zero counts (i.e., where $x_{\textbf{i}}>0$). Notice that $\Gamma$ can always be found in practice and, when $\boldsymbol{\mathscr{X}}$ is only corrupted by false zeros, $\Gamma\subseteq\Omega$ consists of true non-zero counts. We now proceed to the main result, comparing oracle estimators and our proposed estimator that applies the zero-truncated log-likelihood function (\ref{likelihood0}).

\begin{theorem}
	\label{simplethm}
	Let $I \coloneqq \max_n\{I_n\}$, $\boldsymbol{\mathscr{M}}\in S_{R}^{+}(\beta,\alpha)$, and $\Omega$ be a subset of multi-indices selected uniformly at random from all subsets of the same cardinality. Suppose $\boldsymbol{\mathscr{X}}\in\mathbb{Z}_{+}^{I_1\times\cdots\times I_N}$ is a random tensor with each entry in $\Omega$ generated independently as in (\ref{poisson}) and let $\Gamma\subseteq\Omega$ contain the indices of the non-zero entries of $\boldsymbol{\mathscr{X}}$ restricted to $\Omega$. Then the following statements hold with probability no less than $1-4|\Omega|^{-1}$ when $\min_n\{I_n\} \geq (N-1)\log_2^2\left(\max_n\{I_n\}\right) + 1$:
	\begin{align}
		&\mbox{If } \ \widehat{\boldsymbol{\mathscr{M}}} \in S_{R}^{+}(\beta,\alpha)\mbox{ is such that } 	f_{\Omega}\left(\widehat{\boldsymbol{\mathscr{M}}},\boldsymbol{\mathscr{X}}\right) \geq f_{\Omega}(\boldsymbol{\mathscr{M}},\boldsymbol{\mathscr{X}}),\mbox{ then } \frac{\|\boldsymbol{\mathscr{M}}-\widehat{\boldsymbol{\mathscr{M}}}\|^2}{\|\boldsymbol{\mathscr{M}}\|^2} \leq \epsilon. \label{errorbounds_intro1}\\
		&\mbox{If } \ \widetilde{\boldsymbol{\mathscr{M}}}\in S_{R}^{+}(\beta,\alpha)\mbox{ is such that } \tilde{f}_{\Gamma}\left(\widetilde{\boldsymbol{\mathscr{M}}},\boldsymbol{\mathscr{X}}\right) \geq \tilde{f}_{\Gamma}(\boldsymbol{\mathscr{M}},\boldsymbol{\mathscr{X}}),\mbox{ then } \frac{\|\boldsymbol{\mathscr{M}}-\widetilde{\boldsymbol{\mathscr{M}}}\|^2}{\|\boldsymbol{\mathscr{M}}\|^2} \leq \, \kappa \epsilon , \label{errorbounds_intro}
	\end{align}
	where 
 \[
 \epsilon = \mathcal{O}\left(\frac{R\sqrt{I}\log_2(I)}{\sqrt{|\Omega|}}\right)
 \]
and 
 \begin{align} \label{def-kappa}
    \kappa \coloneqq \frac{(4+\beta\tau)e^{\beta}-4}{2(e^{\beta}-\beta-1)} \ \text{ with } \ \tau \coloneqq \frac{1}{\alpha(e^2-2) + 3\log_2(|\Omega|)}. 
 \end{align}
\end{theorem}
The result states that if the number of true counts $|\Omega|$ is proportional to $IR^2\log_2^2(I)$, then estimators with relatively large likelihoods (\ref{likelihood}) and (\ref{likelihood0}) are accurate approximations of the true data model. Furthermore, our approach that applies the zero-truncated likelihood function on the subset of non-zero counts ($\Gamma$) is subject to an error amplification term $\kappa\geq 1$ that depends on $\beta,\alpha,$ and $|\Omega|$. The proof is postponed until Appendix~\ref{sec:proofs}, where Theorem~\ref{simplethm} results from combining Theorems~\ref{mainthm0} and~\ref{mainthm}. To further develop the implications of the result, we narrow down the context to specify our approach and compare it with the ideal oracle scenario mentioned before. 

Suppose our given count data  $\boldsymbol{\mathscr{X}}$ is corrupted by false zeros but otherwise possesses true non-zero counts. Let us further suppose that an oracle provides us with $\Omega$ specifying all true counts obeying (\ref{poisson}). Notice that $\Omega$ contains all non-zeros along with true zero counts, which we assume are distributed in a random manner. Then $\Gamma$ is simply the set of all non-zero entries of $\boldsymbol{\mathscr{X}}$, which can always be identified in practice regardless of $\Omega$. However, in this non-oracle scenario, $\Omega$ still plays an important role (albeit implicitly) since it determines the degree of false zero-congestion in our observations.

To be concrete, let us produce our estimators via maximum likelihood
\begin{equation}
	\label{opt}
	\widehat{\boldsymbol{\mathscr{M}}} = \argmax_{\boldsymbol{\mathscr{T}}\in S_{R}^+(\beta, \alpha)}f_{\Omega}(\boldsymbol{\mathscr{T}},\boldsymbol{\mathscr{X}}) \ \ \ \ \mbox{and} \ \ \ \ \widetilde{\boldsymbol{\mathscr{M}}} = \argmax_{\boldsymbol{\mathscr{T}}\in S_{R}^+(\beta, \alpha)}\tilde{f}_{\Gamma}(\boldsymbol{\mathscr{T}},\boldsymbol{\mathscr{X}}),
\end{equation}
where $\widehat{\boldsymbol{\mathscr{M}}}$ is the oracle estimator and $\widetilde{\boldsymbol{\mathscr{M}}}$ is the zero-truncated estimator. Each estimator satisfies \eqref{errorbounds_intro1} and \eqref{errorbounds_intro}, respectively. 
Theorem~\ref{simplethm} implies that, in contrast to the accuracy of $\widehat{\boldsymbol{\mathscr{M}}}$, the error of our proposed estimator $\widetilde{\boldsymbol{\mathscr{M}}}$ is possibly amplified by $\kappa$ given by \eqref{def-kappa}, which satisfies the inequalities
\begin{equation}
\label{kappa}
1 < \kappa < \infty \ \text{ for } \ \beta > 0.
\end{equation}
The parameter $\kappa$ is a function of tensor search space bounds $\alpha, \beta$ \eqref{t-search-space} and the sample size $|\Omega|$.
A straight-forward analysis shows that with fixed $\alpha$ and $ |\Omega|$, the amplification $\kappa$ increases monotonically with decreasing $\beta$, which is a lower bound for the Poisson parameters. The bound \eqref{errorbounds_intro} and $\kappa$ can be informative to determine in which cases the zero-truncated approach can lead to large errors relative to the oracle estimator. A small $\beta$ implies that the number of true zero counts neglected may be significant and our proposed estimator will likely degrade in accuracy as a consequence. Our zero-truncated approach will not be efficient in small Poisson parameter regimes, but otherwise performs nearly as well as the oracle estimator. These observations will be explored numerically in Section~\ref{sec:experiments}. 

For low-rank tensors, the number of true counts required by the result for an accurate estimator is small relative to the ambient dimensions, i.e., $|\Omega|\sim IR^2\log_2^2(I) \ll \prod_kI_k$. This allows for statistical inference via multiway analysis under significantly under-determined scenarios, which otherwise would require the entire volume to be observed in a setting free of false zeros. Theorem \ref{simplethm} is slightly pessimistic since the optimal sampling rate for elements of $S_{R}^{+}(\beta, \alpha)$ is conjectured to be $|\Omega|\sim IR\log(I)$, where the logarithmic term is unavoidable in matrix and tensor completion under random sampling models \cite{5452187}. Despite this, our derived sampling complexity is novel in that it improves upon current results in the literature, which involve super-quadratic dependence on $R$ and $I$ for $N$-way arrays with $N\geq 3$ \cite{nuctensor,NIPS2013_2050e03c,NEURIPS2019_a1519de5}. However, it is important to notice that we consider the NNCP rank rather than the general CP rank so that this comparison is difficult to make fairly. See Section \ref{relatedwork} for further discussion on the novelty of the result and comparison to other work in the literature.

Theorem \ref{simplethm} does not provide a method for parameter estimation and instead assumes an estimator $\widetilde{\boldsymbol{\mathscr{M}}}$ is available. We state the result in this abstract manner in order to remain flexible and practical. Indeed, outputs of the form (\ref{opt}) are NP-hard to compute \cite{NP}, so that no tractable algorithm is guaranteed to achieve the global optimizer. For this reason we do not specify how $\widetilde{\boldsymbol{\mathscr{M}}}$ should be produced and instead attempt to state minimal conditions that an accurate estimate should satisfy, in order to guide practitioners into developing appropriate methods. In fact, the result only requires for an estimator to have large likelihood relative to the true parameter tensor. Therefore, a global optimum of (\ref{opt}) may not be needed and the result remains informative to local optima and other less greedy methods. In Section \ref{sec:experiments}, we explore the theoretical observations of this section numerically.

\subsection{Connections with Prior Work and Innovations}
\label{relatedwork}

Our work falls within the vast literature of matrix and tensor completion, see, e.g., \cite{5452187,5466511,navid1,PoissonMC,nuctensor,Gandy_2011,NIPS2014_c15da1f2,montanari,7973019} with the paper \cite{PoissonMC} closest to our work. We generalize the Poisson matrix completion approach and result of \cite{PoissonMC} to larger dimensional arrays and the setting of false zero corruption. In the case of two dimensions, our derived sampling complexity is worse than the near-optimal matrix completion result of \cite{PoissonMC} due to our quadratic dependence upon the rank $|\Omega|\sim IR^2\log^2(I)$. However, for general $N$-way arrays with $N\geq 3$ no result exists exhibiting the theoretic sample complexity rate $|\Omega|\sim IR\log(I)$ \cite{Lee,NNTC} and our result improves upon the literature in this regard. 

The main results in the tensor completion literature provide general $N$-way array sampling complexities $\sim I^{N/2}R\polylog(I)$ \cite{montanari}, $\sim (I^{3/2}R^{(N-1)/2} + IR^{N-1})\log^2(I)$ \cite{7973019}, and $\sim IR^{3N-3}\log^2(I)$ \cite{navid1,navid2,navid3}. Notice that the dependence of these rates on the rank or largest array dimension is polynomial in terms of $N$. Our main results are able to provide sampling complexity $|\Omega|\sim IR^2\log^2(I)$, which is independent of $N$ (exponentially) and nearly matches the optimal rate. However, we stress that our work applies to nonnegative tensor decompositions (NCCP). This context is crucial for our sampling complexity, which complicates a fair comparison of our work to the citations discussed. If we consider the general CP rank, our derived sampling complexity matches the results in \cite{navid1,navid2,navid3} (see Section \ref{sec:proofs}). The contribution of our work is to show that the proof technique of \cite{navid1,navid2,navid3} can remove the exponential dependency of $I$ and $R$ on $N$ when one considers the NNCP rank. We note that the work \cite{NNTC} also exploits nonnegative tensors to obtain $|\Omega|\sim IR^4\log^2(I)$, but the result does not allow for an arbitrarily small error bound.

Focusing on literature related to our proposed approach, the papers \cite{ZTP1,ZTP2} also consider the utility of zero-truncated distributions to appropriately disregard zero values. Therein, the authors ignore zero values to reduce the amount of data to be processed and efficiently scale their inference procedure to large dimensional volumes. Our approach also scales to large dimensions, but our main focus is to filter out the corrupt portion of the data and provide inference error bounds.

%% file: experiments.tex
We present a series of experiments to illustrate the influence of several problem parameters on Theorem~\ref{simplethm} in practice. Specifically, we demonstrate the errors associated with the estimators $\widehat{\boldsymbol{\mathscr{M}}}$ and $\widetilde{\boldsymbol{\mathscr{M}}}$ with respect to $\boldsymbol{\mathscr{M}}$ when these estimators are computed using the method of maximum likelihood estimation. These experiments illustrate some of the practical ramifications of Theorem~\ref{simplethm}.

\subsection{Experimental Data}
\label{sec:exp_data}
We generate synthetic data using the approach first described by Chi and Kolda in~\cite{Chi2012}, which is implemented in the Tensor Toolbox for MATLAB~\cite{TTB_src} in the method \texttt{create\_problem}. We generate random instances of $N$-way tensors $\boldsymbol{\mathscr{M}}$, with all dimensions of size $I$, having rank-$R$ multilinear structure as represented in the CP model:

\begin{equation} 
	\boldsymbol{\mathscr{M}} = \llbracket 
	\boldsymbol{\lambda}; 
	\mathbf{A}^{(1)},\ldots,\mathbf{A}^{(N)}\rrbracket = \sum_{r=1}^R \lambda_r 
	\mathbf{a}_r^{(1)} \circ \ldots \circ \mathbf{a}_r^{(N)} \; ,\label{eq:cpmodel}
\end{equation}
where $\mathbf{A}^{(n)} \in \mathbb{R}^{I \times R} \; \forall n \in [N]$.

 We create the desired low-rank, multilinear structure such that all of the entries in $\boldsymbol{\mathscr{M}}$ lie in the interval $[\beta, \alpha]$, as prescribed in Theorem~\ref{simplethm} via a sampling of the entries in the factor matrices, $\mathbf{A}^{(1)},\ldots,\mathbf{A}^{(N)}$, uniformly from $[(\beta/R)^{1/N}, (\alpha/R)^{1/N}]$, and set $\lambda_r = 1$.
The result is that the entries in $\boldsymbol{\mathscr{M}}$ follow a truncated normal distribution in the interval $[\beta, \alpha]$. Figure~\ref{fig:Mdist} illustrates the distribution of entries of an instance of $\boldsymbol{\mathscr{M}}$ generated using $\beta=1.5$ and $\alpha=2.5$. 


\begin{figure}[!t]
	\centering
	\begin{tabular}{cc}
		$I=50$ &  $I=200$ \\
		\includegraphics[width=0.45\textwidth]{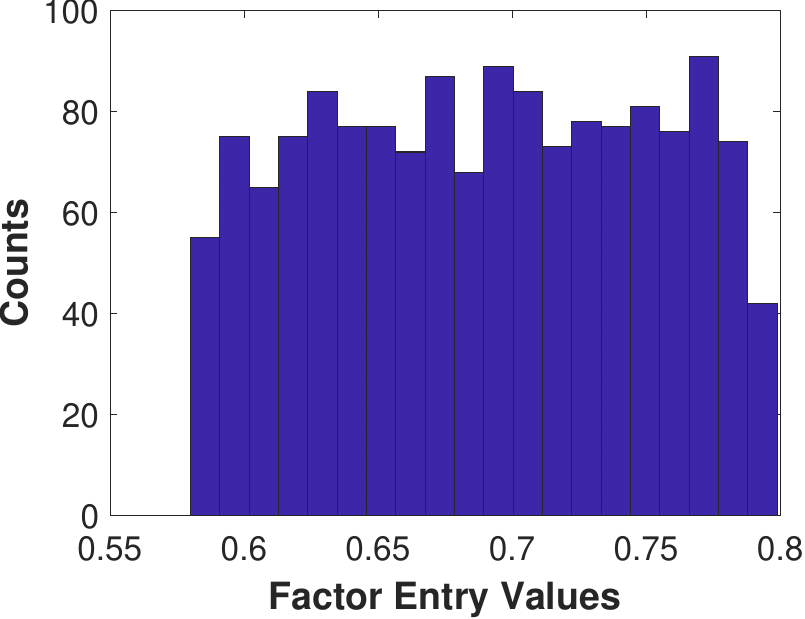}&
		\includegraphics[width=0.45\textwidth]{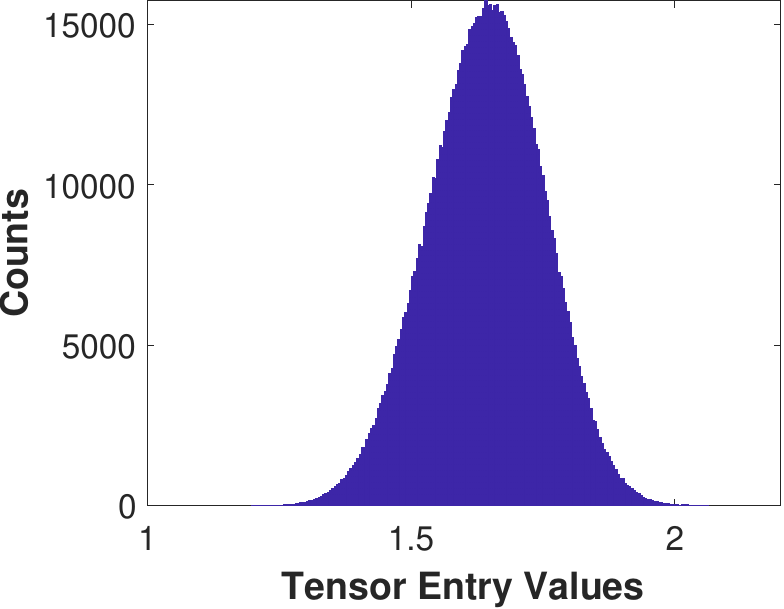}\\
	\end{tabular}
	\caption{Histograms of entries of example factor matrices $\mathbf{A}^{(1)},\ldots,\mathbf{A}^{(N)}$ (left) and tensor $\boldsymbol{\mathscr{M}}$ (right) generated via {\normalfont \texttt{create\_problem}} with $\beta=1$, $\alpha=2.5$, $N=3$, $I=100$, and $R=5$.}
	\label{fig:Mdist}
\end{figure}

We generate instances of $\boldsymbol{\mathscr{X}}$ by first creating an instance of $\boldsymbol{\mathscr{M}}$ using the procedure above, and then use the Poisson random sampler, \texttt{poissrnd}, from MATLAB's Statistics and Machine Learning Toolbox, to generate the entries of $\boldsymbol{\mathscr{X}}$.

Instances of the index set $\Omega$ are constructed by uniformly sampling without replacement from the linearized index set of $\boldsymbol{\mathscr{X}}$, given by $[I^N]$. Thus, when simulating false zeros in $\boldsymbol{\mathscr{X}}$, the values at the indices in $[I^N] \setminus \Omega$ are set to $0$.

\subsection{Maximum Likelihood Estimation Methods}
\label{sec:exp_mle}
Given a data tensor $\boldsymbol{\mathscr{X}}$ whose entries are each assumed to be a draw from a Poisson distribution with parameters in $\boldsymbol{\mathscr{M}}$, as defined in \eqref{poisson}, we compute estimators for $\boldsymbol{\mathscr{M}}$ using the method of maximum likelihood estimation~\cite{MYUNG200390}. We solve the maximum likelihood estimation problem by minimizing the negative of the log-likelihood function associated with the distributions of interest. Specifically, in our experiments, we minimize $-f_{\Omega}(\boldsymbol{\mathscr{M}}, \boldsymbol{\mathscr{X}})$ from \eqref{likelihood} and $-\tilde f_{\Gamma}(\boldsymbol{\mathscr{M}}, \boldsymbol{\mathscr{X}})$ from \eqref{likelihood0} to compute estimators $\widehat{\boldsymbol{\mathscr{M}}}$ and $\widetilde{\boldsymbol{\mathscr{M}}}$, respectively. In other words, we attempt to solve (\ref{opt}) without constraining estimator entries to lie in $[\beta,\alpha]$. Note that this range for the estimated Poisson parameters is required for Theorem \ref{simplethm}. However, we choose to conduct our experiments under the more realistic scenario that such bounds are not known or implemented. Our numerical results in Section \ref{sec:exp_results} will demonstrate that unconstrained estimation produces accurate estimates that illustrate our theoretical statements in a practical setting.

The Generalized Canonical Polyadic (GCP) method for computing low-rank CP decompositions~\cite{HongKoldaDuersch2020, Kolda2020} provides a method for maximum likelihood estimation using general loss functions that we use here in our experiments. Specifically, we use the Tensor Toolbox for MATLAB implementation of GCP, provided in the method \texttt{gcp\_opt}, to compute maximum likelihood estimators for $\boldsymbol{\mathscr{M}}$. In \texttt{gcp\_opt}, we use the limited-memory bound-constrained quasi-Newton optimization method~\cite{lbfgsb_c, lbfgsb}; i.e., the input parameter \texttt{opt} is set to \texttt{'lbfgsb'}.
	
We compute three estimators denoted \emph{Poisson}, \emph{Oracle}, and \emph{ZTP}:
\begin{itemize}
	\item \textbf{\emph{Poisson}}. This  approach was introduced in~\cite{Chi2012} for computing CP decompositions of data tensors with count values. It computes an estimate by minimizing $-f_{\Omega}(\boldsymbol{\mathscr{M}}, \boldsymbol{\mathscr{X}})$ over all values in $\boldsymbol{\mathscr{X}}$, i.e., by setting $\Omega = [I^N]$. Thus, it treats both true and false zeros as zero values in the data. In \texttt{gcp\_opt}, the input parameter \texttt{type} is set to \texttt{'count'} to specify this method.
	\item \textbf{\emph{Oracle}}. This approach is similar to the \emph{Poisson} method except that the estimate uses only the true zeros and non-zeros in $\boldsymbol{\mathscr{X}}$. Thus, the estimate ignores the zeros values in $\boldsymbol{\mathscr{X}}$ that correspond to false zeros by removing the indices of the false zeros from $\Omega$. In general, this information about the specific types of zero values in a data tensor is unknown. However, since we generate $\Omega$ in our experiments, this information is known explicitly. Thus, we can use this estimate when zeros in data are known to be true or false \emph{a priori}. In \texttt{gcp\_opt}, the input parameter \texttt{mask} is set to be a tensor of the same size of $\boldsymbol{\mathscr{X}}$ whose values at indices in $\Omega$ are equal to $1$ and $0$ otherwise. This provides the information to GCP to minimize only over the true zeros and non-zeros in $\boldsymbol{\mathscr{X}}$ when computing an estimator. All other input parameters are the same as those used for the \emph{Poisson} method. 
	\item \textbf{\emph{ZTP}}. This approach computes an estimate by minimizing $-\tilde f_{\Gamma}(\boldsymbol{\mathscr{M}}, \boldsymbol{\mathscr{X}})$, where $\Gamma \subseteq \Omega$ denotes the indices of the non-zeros of $\boldsymbol{\mathscr{X}}$. Thus, no zero values are used in computing an estimator with this method, which is accounted for in the zero-truncated Poisson log-likelihood function, defined in~\eqref{likelihood0}. In \texttt{gcp\_opt}, the input parameters \texttt{func} and \texttt{grad} are set to anonymous function handles for code to compute $-\tilde f_{\Gamma}(\boldsymbol{\mathscr{M}}, \boldsymbol{\mathscr{X}})$ and $- \nabla \tilde f_{\Gamma}(\boldsymbol{\mathscr{M}}, \boldsymbol{\mathscr{X}})$, respectively. As for the \emph{Oracle} method, the input parameter \texttt{mask} is set to be a tensor of the same size of $\boldsymbol{\mathscr{X}}$ whose values at indices in $\Gamma$ are equal to $1$ and all other entries are equal to $0$.
\end{itemize}


\subsection{Average Relative Error}
\label{averagerelativerror}

Estimator errors are computed as the relative difference between the estimators and Poisson parameter tensors, as in (\ref{errorbounds_intro1}) and (\ref{errorbounds_intro}). For each instance pair $(\boldsymbol{\mathscr{M}},\boldsymbol{\mathscr{X}})$, we report the average relative error (denoted as \emph{Average Relative Error} in the plots presented in \S\ref{sec:exp_results}) across $k$ randomly selected instances of the index set $\Omega$. 
Table~\ref{tab:error} presents the maximum likelihood estimate (MLE) methods, the indices of entries in  $\boldsymbol{\mathscr{X}}$ used for each MLE method, and the corresponding relative error expressions. 
Note that estimators $\widehat{\boldsymbol{\mathscr{M}}}$ and $\widetilde{\boldsymbol{\mathscr{M}}}$ are those computed using the Poisson log-likelihood~\eqref{likelihood} and zero-truncated Poisson log-likelihood~\eqref{likelihood0} functions, respectively.

\begin{table}[h!]
	\renewcommand*{\arraystretch}{1.33}
	\caption{Data indices relative error expressions used for the MLE methods in experiments.}\label{tab:error}
        \centering
	\begin{tabular}{|c|c|c|}
		\hline
		\textbf{\emph{MLE Method}} & \textbf{\emph{Data Indices}} & \textbf{\emph{Relative Error}}\\
		\hline
		\emph{Poisson} & $[I^N]$ & $\|\boldsymbol{\mathscr{M}}-\widehat{\boldsymbol{\mathscr{M}}} \| / 
		{\|\boldsymbol{\mathscr{M}}\|}$\\
		\hline
		\emph{Oracle} & $\Omega$ & $\|\boldsymbol{\mathscr{M}}-\widehat{\boldsymbol{\mathscr{M}}} \| / {\|\boldsymbol{\mathscr{M}}\|}$\\
		\hline
		\emph{ZTP} & $\Gamma$ & $\|\boldsymbol{\mathscr{M}}-\widetilde{\boldsymbol{\mathscr{M}}} \| / {\|\boldsymbol{\mathscr{M}}\|}$\\
		\hline
	\end{tabular}
	
\end{table}

\subsection{Experimental Setup}
\label{sec:exp_setup}
Our experiments illustrate the differences in computing maximum likelihood estimators for $\boldsymbol{\mathscr{M}}$ using the various methods described in \S\ref{sec:exp_mle}. Specifically, in these experiments, we vary the size of the number of trusted data tensor entries, $|\Omega|$, and the ranges of the Poisson parameter tensor entries, $[\beta,\alpha]$.

We run several experiments by varying $|\Omega|$, $\beta$, and $\alpha$. In all experiments, we use $N=3$ and $R=5$. Since the minimum requirement for each dimension of these experiments is $I \geq 82$, as specified in the setup of Theorem \ref{simplethm}, we use values of $I \in \{50, 100, 200\}$ to illustrate the impact of dimension size on the results. For each experiment, we use $\beta$ and $\alpha$ to generate instances of $\boldsymbol{\mathscr{M}}$ and $\boldsymbol{\mathscr{X}}$ as described in \S\ref{sec:exp_data}. For each instance pair of $(\boldsymbol{\mathscr{M}}, \boldsymbol{\mathscr{X}})$, we generate $k=50$ instances of $\Omega$. Also, due to the nonconvexity of the negative log-likelihood functions being minimized, we compute estimators for each instance of $\Omega$ starting from $n=20$ initial starting points.

Across the experiments, we vary the problem parameters $|\Omega|$, $\beta$, and $\alpha$ as follows.

\begin{itemize}
	\item \emph{\textbf{Varying $\boldsymbol{|\Omega|}$}}. We vary the size of the set of true zero and non-zero  values, $|\Omega|$, such that $|\Omega|/I^N$ falls in the range $[0,1]$. Results for the different methods are reported as a function of $|\Omega|/I^N$, even though different amounts of data are used in computing the estimators with the different methods, as discussed in \S\ref{sec:exp_mle}. We run experiments with  $|\Omega|/I^N \in \{0.01, 0.02, 0.03, 0.04, 0.05, 0.10, 0.15, \dots, 0.95, 1.0\}$.
	\item \emph{\textbf{Varying $\boldsymbol{\beta}$}}. The probability of generating true zeros in $\boldsymbol{\mathscr{X}}$ increases as $\beta \to 0$. Since the different estimator methods treat zeros differently, it is important to understand the impact of the number of true zeros in $\boldsymbol{\mathscr{X}}$ on the estimator errors. We run experiments with $\beta \in \{0.001, 0.01, 0.1, 1\}$.
	\item \emph{\textbf{Varying $\boldsymbol{\alpha}$}}. The probability of generating true zeros in $\boldsymbol{\mathscr{X}}$ decreases with increasing  $\alpha$. When there are no true zeros in $\boldsymbol{\mathscr{X}}$, the Oracle and ZTP methods are equivalent. Moreover, when there are no true or false zeros in $\boldsymbol{\mathscr{X}}$---i.e., when $|\Omega| = I^N$---all three methods described in \S\ref{sec:exp_mle} are equivalent. We run experiments with $\alpha \in \{2.5, 5, 10, 25, 50\}$.
\end{itemize}

\subsection{Implementation Details}
\label{sec:exp_impl}
See Appendix~\ref{sec:impl} for implementation details of the experiments described in Sections~\ref{sec:exp_data}--\ref{sec:exp_setup}. 

\subsection{Results}
\label{sec:exp_results}
We present results for experiments involving the methods defined as \emph{Poisson}, \emph{Oracle}, and \emph{ZTP} in \S\ref{sec:exp_mle} to demonstrate the results of Theorem~\ref{simplethm} in practice. 

\emph{\textbf{Varying $\boldsymbol{|\Omega|}$}}. Figure~\ref{fig:varying_omega} presents the average relative errors of estimators using the three methods as a function of $|\Omega|/I^N$, which is the fraction of the number true zeros and non-zeros to the total number of entries in the data tensors. In these experiments, we set $\beta=1$, $\alpha=2.5$, $N=3$, $I=100$, $R=5$, generate 50 replicates of $\Omega$ for each value of $|\Omega|/I^N$, and compute estimators using the different methods starting from $n=20$ randomly generated initial starting points for each instance of $\Omega$. As expected, the \emph{Oracle} method, which only computes estimators using true zeros and non-zeros, leads to the best results for all values of $|\Omega|/I^N$. When $|\Omega|/I^N = 1$, the \emph{Poisson} and \emph{Oracle} methods are identical, since there are no false zeros, as illustrated in the right side of the plot. In such cases, though, the \emph{ZTP} method ignores all zeros and thus incurs more error in the estimates. As predicted by Theorem~\ref{simplethm}, we see that the average errors of the \emph{ZTP} estimators track those of the \emph{Oracle} estimators, differing only by a small multiplicative value at each value of $|\Omega|/I^N$. In these experiments, the predicted difference in relative error in Theorem~\ref{simplethm} should be bounded by a factor of $\sqrt{\kappa}$, which aligns well with the results presented in Figure~\ref{fig:varying_omega}. These results are very consistent across the $k=50$ replicates of $\Omega$ and the $n=20$ randomly generated initial starting points of the numerical optimization methods used. Specifically, the shaded regions in Figure~\ref{fig:varying_omega} represent one standard deviation away from the average relative errors for each method across the replicates. Furthermore, the standard deviations in relative error are all more than two orders of magnitude smaller on average across the initial starting guesses than those for the replicates. Together, these results indicate very little variability in the estimators computed using all three methods.

\begin{figure}[b!]
    \centering
    \includegraphics[width=.5\textwidth]{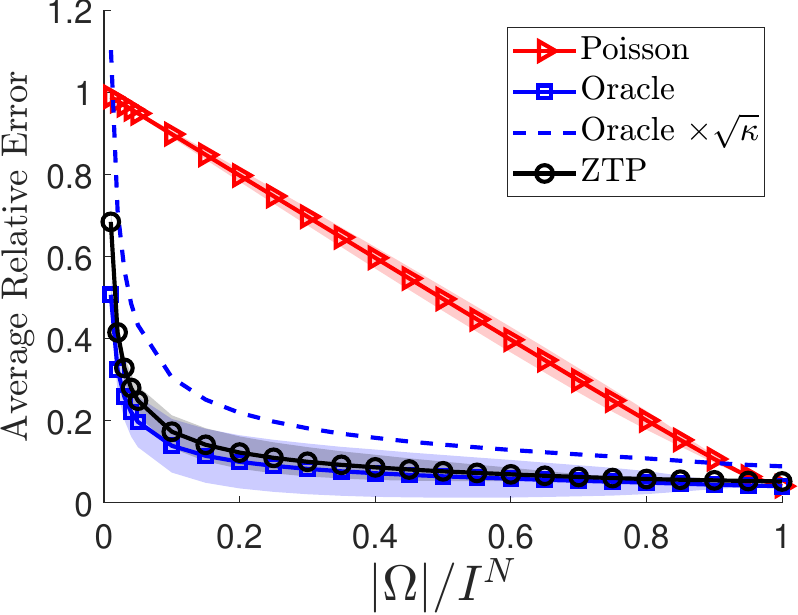}
	\caption{Results varying $|\Omega|$: $\beta=1$, $\alpha=2.5$, $I=100$, $N=3$, $R=5$, and 50 replicates. The solid lines represent the mean errors across the 50 replicates, and the shaded regions represent the one standard deviation away from the mean errors.}
	\label{fig:varying_omega}
\end{figure}

\begin{figure}[ht!]
    \centering
	\begin{tabular}{cc}
		$\beta=1$ &  $\beta=0.100$ \\
		\includegraphics[width=.4\textwidth]{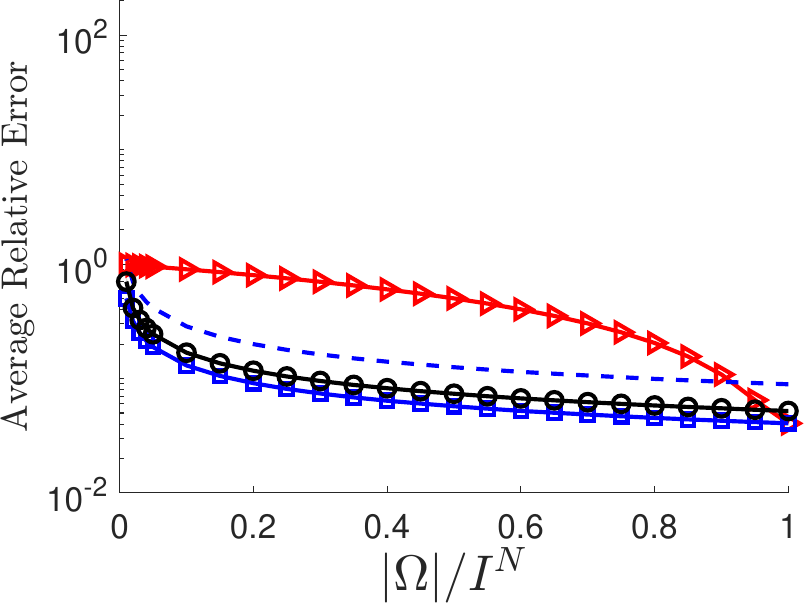}&
		\includegraphics[width=.4\textwidth]{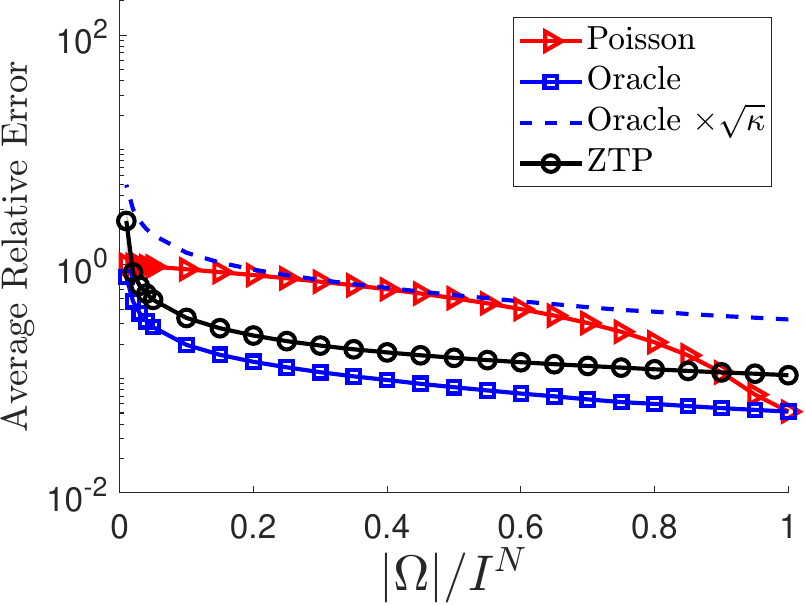}\\[.1in]
		$\beta=0.010$ &  $\beta=0.001$ \\
		\includegraphics[width=.4\textwidth]{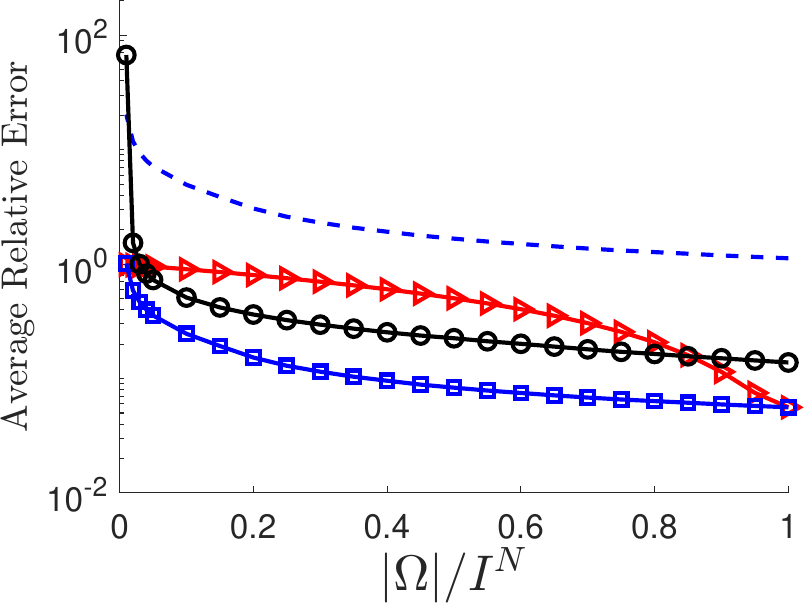}&
		\includegraphics[width=.4\textwidth]{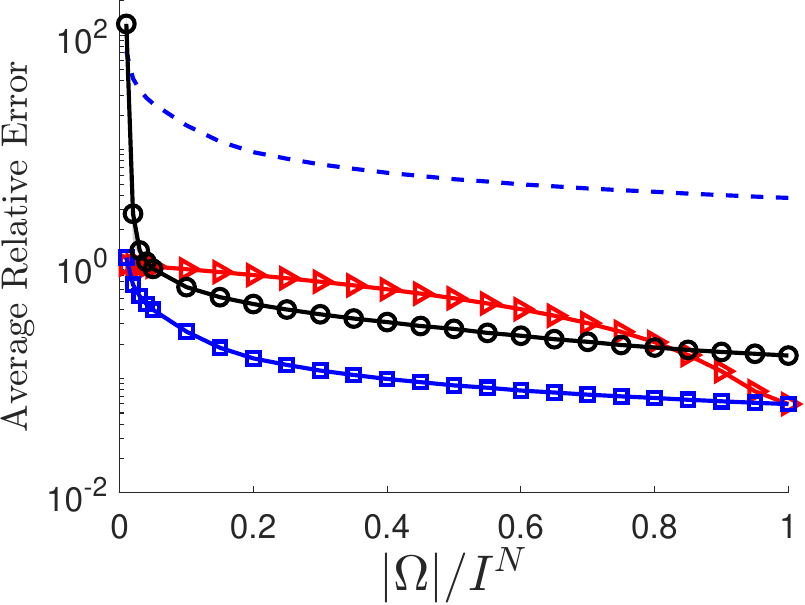}
	\end{tabular}
	\caption{Results varying $\beta$: $\alpha=2.5$, $I=100$, $N=3$, $R=5$, and 50 replicates.}
	\label{fig:varying_beta}
\end{figure}

\emph{\textbf{Varying $\boldsymbol{\beta}$}}. Figure~\ref{fig:varying_beta} presents the average relative errors of estimators using the three methods as a function of $\beta$, which influences the number of true zeros in the data tensors. As expected, as $\beta \to 0$, $\kappa$ increases, and thus there are greater differences in the average errors between the estimators computed with the \emph{Oracle} and \emph{ZTP} methods. Moreover, these differences are much more extreme as $|\Omega|/I^N \to 0$---i.e., as the numbers of false zeros in the data tensors increase. When $\beta$ is close to 0, there are few observations used by the \emph{ZTP} method to compute the estimator, and thus we see that the average relative errors can be large, whereas the average relative errors for the \emph{Oracle} method are still bounded by the results of computing estimators using the \emph{Poisson} method. Thus, we recommend that the \emph{ZTP} method be used only when there are a sufficient number of non-zero entries in the data tensors; the specific fractions will be determined by the number of dimensions, sizes of those dimensions, and the distributions of values of the non-zero entries.

\emph{\textbf{Varying $\boldsymbol{\alpha}$}}. Figure~\ref{fig:varying_alpha} presents the average relative errors of estimators using the three methods as a function of $\alpha$, which also influences the number of true zeros in the data tensors. We see that for fixed values of $\beta$ (in this case $\beta=0.1$), as $\alpha$ increases, there is very little difference in average relative errors between estimators computing using the \emph{Oracle} and \emph{ZTP} methods. These results are due to the fact that as $\alpha$ increases, the probability of generating true zeros in the data tensors decreases. Thus, with fewer true zeros, the differences between these methods are diminished.

\emph{\textbf{Varying $\boldsymbol{I}$}}. Figure~\ref{fig:varying_dimensions} presents the average relative errors of estimators using the three methods for values of $I \in \{50, 200\}$, which represents smaller and much larger dimension sizes than those required for the results in Theorem~\ref{simplethm}. For the results presented here, $\beta=1$ and $\alpha=2.5$.
Recall that when $N=3$ and $R=5$, we require that $I \geq 82$ for the results in Theorem~\ref{simplethm} to hold. We see that when this requirement is not satisfied---e.g., when $I=50$---the average relative errors are worse than expected, with rapid increases as $|\Omega|/I^N \to 0$. Alternatively, as $I$ increases well above the minimum value required to support the conclusions of Theorem~\ref{simplethm}---e.g., when $I=200$---we see that both the \emph{Oracle} and \emph{ZTP} methods produce even better results in terms of average relative errors for the estimators computed. Since the relative errors in Theorem~\ref{simplethm} are functions of $I$ for fixed values of $\beta$, $\alpha$, $N$, $R$, and $|\Omega|$, these results indicate good agreement between theory and practice.

\begin{figure}[ht!]
	\centering
	\begin{tabular}{cc}
		$\alpha=5$ &  $\alpha=10$ \\
		\includegraphics[width=.4\textwidth]{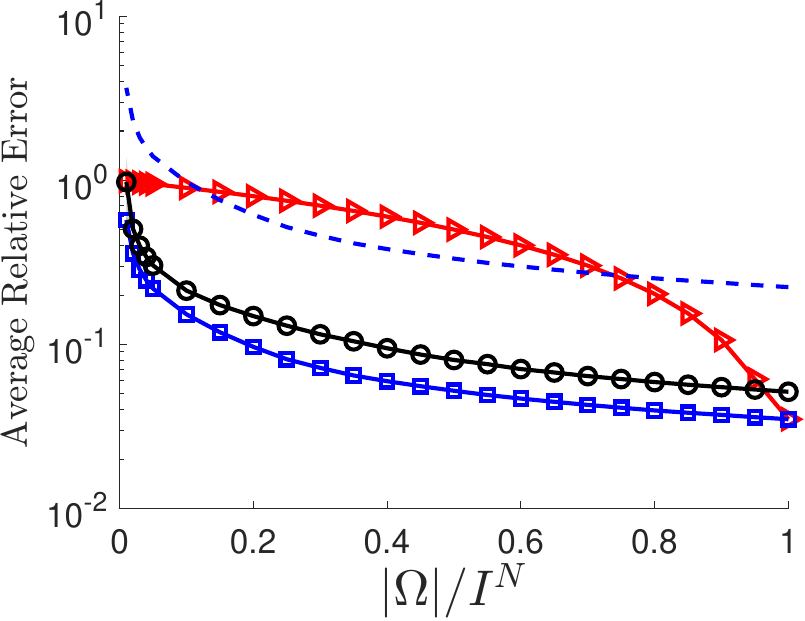}&
		\includegraphics[width=.4\textwidth]{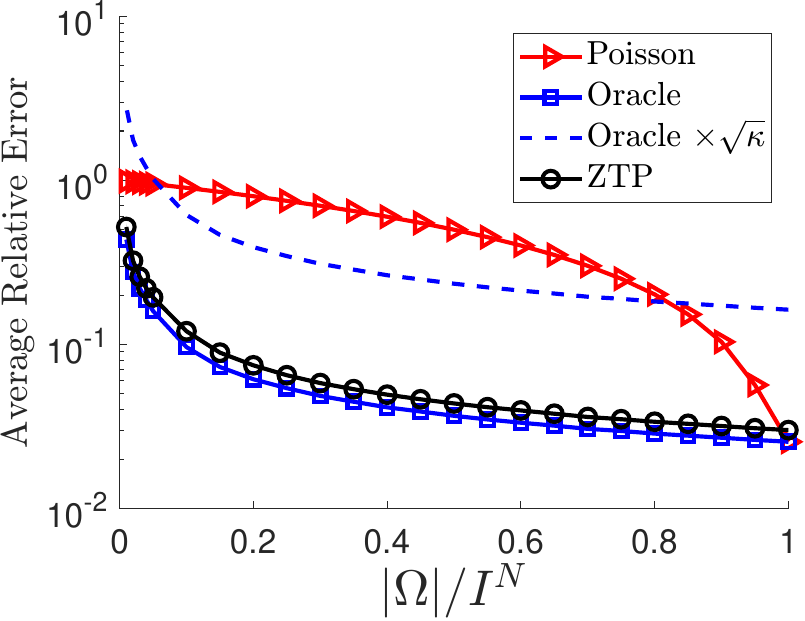}\\
		$\alpha=25$ &  $\alpha=50$ \\
		\includegraphics[width=.4\textwidth]{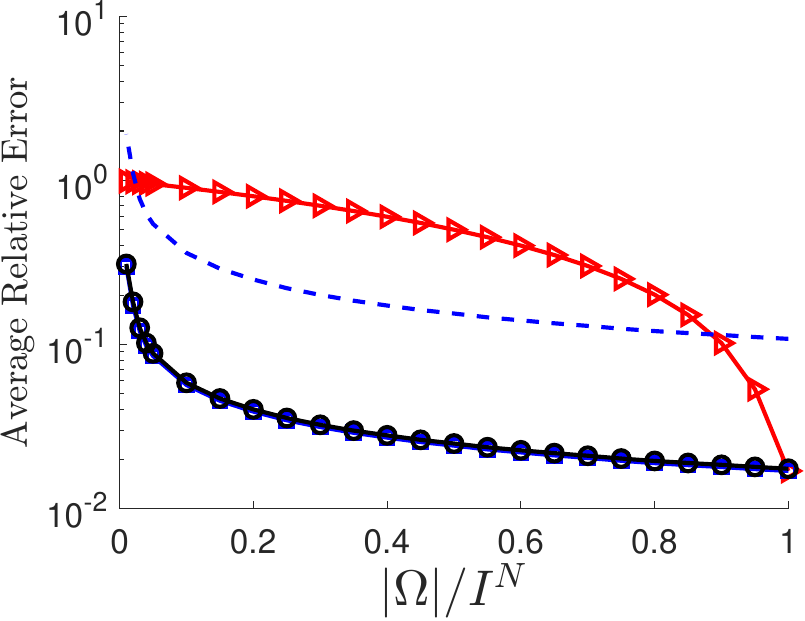}&
		\includegraphics[width=.4\textwidth]{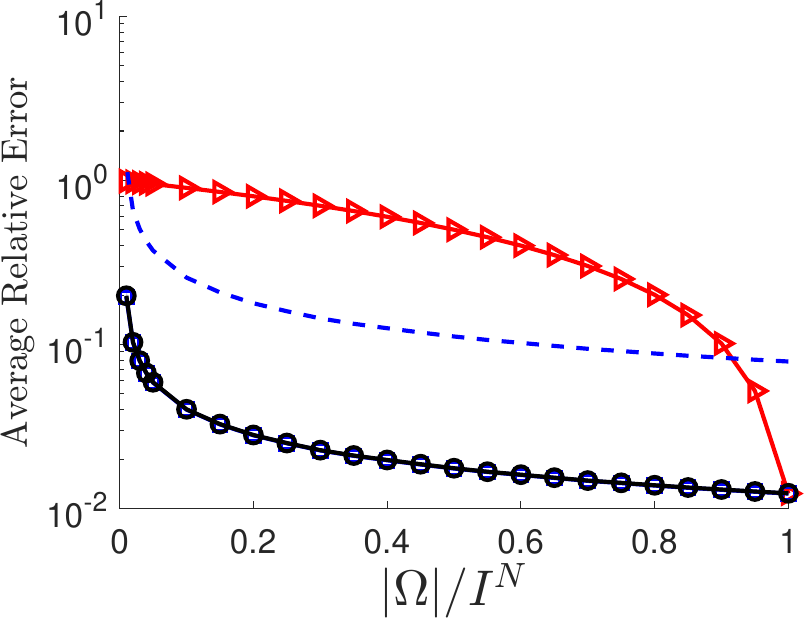}
	\end{tabular}
	\caption{Results varying $\alpha$: $\beta=0.1$, $I=100$, $N=3$, $R=5$, and 50 replicates.}
	\label{fig:varying_alpha}
\end{figure}

\begin{figure}[ht!]
	\centering
	\begin{tabular}{cc}
		$I=50$ &  $I=200$ \\
		\includegraphics[width=.4\textwidth]{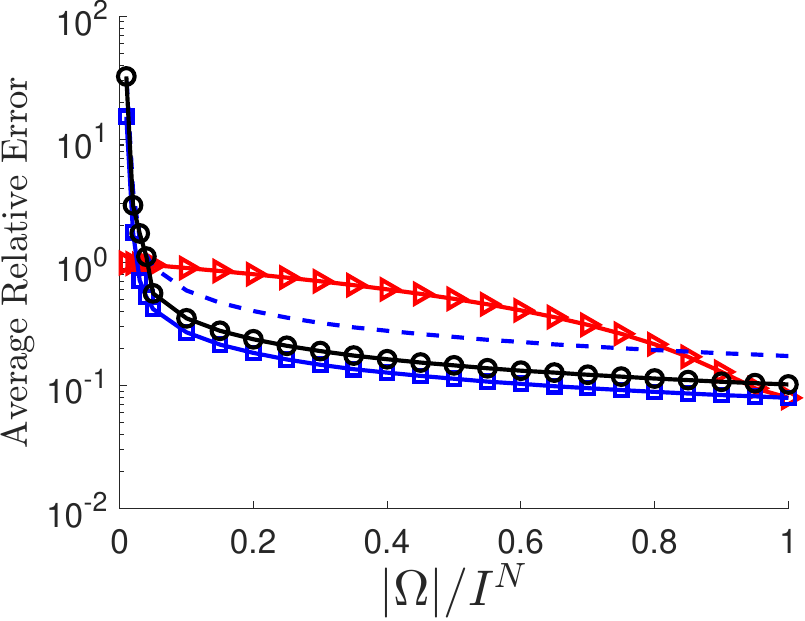}&
		\includegraphics[width=.4\textwidth]{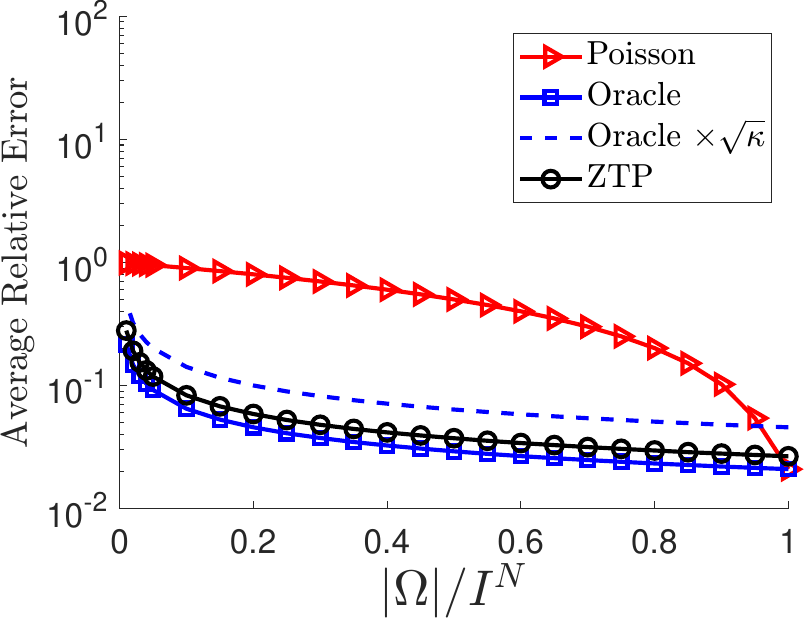}\\
	\end{tabular}
	\caption{Results varying $I$: $\beta=1$, $\alpha=2.5$, $N=3$, $R=5$, and 50 replicates.}
	\label{fig:varying_dimensions}
\end{figure}

%% file: main.tex
We now present the main results for our proposed zero-truncated approach and the ideal Poisson regression methodology (i.e., the oracle estimator). This section includes two theorems that independently provide error bounds for the zero-truncated Poisson estimator $\widetilde{\boldsymbol{\mathscr{M}}}$ (Theorem \ref{mainthm0}) and for the oracle estimator $\widehat{\boldsymbol{\mathscr{M}}}$ (Theorem \ref{mainthm}). Each result derives a worse case relative error for each respective methodology, with the purpose of comparing these two approaches analytically (i.e., comparing our proposed method to the ``ideal'' regression method). This leads to Theorem \ref{simplethm} in the introduction, which is a corollary of the two main results of this section. Theorem \ref{simplethm} simply presents the error bounds of Theorems \ref{mainthm0} and \ref{mainthm} together, under simplified circumstances and gathering common terms. The proof of Theorem \ref{simplethm} will be presented after stating Theorems \ref{mainthm0} and \ref{mainthm}. The proofs of these main results can be found in the following subsections, relying on crucial lemmas to establish the theorems. For brevity, we omit the proofs of the required lemmas until Appendix \ref{sec:proofslemmas}.

For compactness, in this section we modify the log-likelihood functions to
\[
f_{\Omega}(\boldsymbol{\mathscr{M}}) \coloneqq \sum_{\textbf{i}\in\Omega}x_{\textbf{i}}\log\left(m_{\textbf{i}}\right)-m_{\textbf{i}},
\]
and
\[
\tilde{f}_{\Omega}(\boldsymbol{\mathscr{M}}) \coloneqq \sum_{\textbf{i}\in\Omega}x_{\textbf{i}}\log\left(m_{\textbf{i}}\right)-\log\left(\exp(m_{\textbf{i}}) - 1\right),
\]
so that their dependency on the count data $\boldsymbol{\mathscr{X}}$ is implicit and the terms $-\log(x_{\textbf{i}}!)$ are removed. We note that any $\widehat{\boldsymbol{\mathscr{M}}},\widetilde{\boldsymbol{\mathscr{M}}}\in S_{R}^{+}(\beta,\alpha)$ satisfying
\[
f_{\Omega}\left(\widehat{\boldsymbol{\mathscr{M}}}\right) \geq f_{\Omega}(\boldsymbol{\mathscr{M}}) \ \ \ \ \mbox{and} \ \ \ \ \tilde{f}_{\Gamma}\left(\widetilde{\boldsymbol{\mathscr{M}}}\right) \geq \tilde{f}_{\Gamma}\left(\boldsymbol{\mathscr{M}}\right)
\]
will also satisfy the requirements in (\ref{errorbounds_intro1}) and (\ref{errorbounds_intro}). Therefore, this modification does not change the statement and simply serves as a means to compress our proofs.

In the interest of generality, we will also state our results in terms of the CP rank \cite{kolda} defined as
\[
\mbox{rank}(\boldsymbol{\mathscr{T}})\coloneqq \min\Bigg\{R\in\mathbb{N} \ \Big| \ \boldsymbol{\mathscr{T}}=\sum_{r=1}^{R} \textbf{a}_{r}^{(1)}\circ \textbf{a}_{r}^{(2)}\circ\cdots \circ \textbf{a}_{r}^{(N)} \ \ \mbox{with} \ \ \textbf{a}_{r}^{(n)}\in\mathbb{R}^{I_n} \ \ \forall r\in [R], n\in [N] \Bigg\},
\]
which simply removes the nonnegative constraints on the factors. We also define the respective the search space
\[
S_R(\beta,\alpha) \coloneqq \Big\{\boldsymbol{\mathscr{T}}\in\mathbb{R}^{I_1\times \cdots\times I_N} \ | \ \beta\leq t_{\textbf{i}}\leq \alpha \ \ \mbox{and} \ \ \mbox{rank}(\boldsymbol{\mathscr{T}})\leq R\Big\}.
\]
We note that we always have rank$(\boldsymbol{\mathscr{T}})\leq$ rank$_{+}(\boldsymbol{\mathscr{T}})$. 

We first present the main result for our proposed methodology. The following theorem provides the error bound of the estimator $\widetilde{\boldsymbol{\mathscr{M}}}$ from Section \ref{sec:intro}, which achieves zero-truncated Poisson tensor completion using only the set of non-zero counts $\Gamma$.

\begin{theorem}
	\label{mainthm0}
	Suppose $\boldsymbol{\mathscr{M}}\in S_R^+(\beta,\alpha)$ and let $\Omega\subseteq [I_1]\times\cdots\times[I_N]$ be a subset of cardinality $|\Omega|\leq I_1\cdots I_N$, chosen uniformly at random from all subsets of the same cardinality. Let $\boldsymbol{\mathscr{X}}\in \mathbb{Z}_+^{I_1\times\cdots \times I_N}$ be a random tensor, with each entry in $\Omega$ generated independently via (\ref{poisson}) and let $\Gamma\subseteq\Omega$ be the set of nonzero entries of $\boldsymbol{\mathscr{X}}$ restricted to $\Omega$. Further suppose that $\min_n\{I_n\} \geq (N-1)\log_2^2\left(\max_n\{I_n\}\right) + 1$ and define
 \[
 \tau \coloneqq \frac{1}{\alpha(e^2-2) + 3\log_2(|\Omega|)}.
 \] 
 
 Fix $\tilde{R}\in\mathbb{N}$, then for any $\widetilde{\boldsymbol{\mathscr{M}}}\in S_{\tilde{R}}^+(\beta,\alpha)$ such that
	\begin{equation}
		\label{feasible0}
		\tilde{f}_{\Gamma}(\widetilde{\boldsymbol{\mathscr{M}}}) \geq \tilde{f}_{\Gamma}(\boldsymbol{\mathscr{M}}),
	\end{equation}
	we have
	\begin{align}
		\label{bound10}
		\frac{\|\boldsymbol{\mathscr{M}}-\widetilde{\boldsymbol{\mathscr{M}}}\|^2}{\|\boldsymbol{\mathscr{M}}\|^2} \leq \, & \frac{64\alpha\left((4+\beta\tau)e^{\beta}-4\right)}{(e^{\beta}-\beta-1)\beta^3\tau} \left(\frac{(\alpha R+\alpha \tilde{R} + 2)\sqrt{\sum_{n=1}^{N}I_n}}{\sqrt{|\Omega|}}\right) 
	\end{align}
	with probability exceeding $1-\frac{2}{|\Omega|}$. Furthermore, in the general case where $\boldsymbol{\mathscr{M}}\in S_{R}(\beta,\alpha)$ and $\widetilde{\boldsymbol{\mathscr{M}}}\in S_{\tilde{R}}(\beta,\alpha)$ but otherwise under the same assumptions, we have
	\begin{align}
		\label{bound20}
		\frac{\|\boldsymbol{\mathscr{M}}-\widetilde{\boldsymbol{\mathscr{M}}}\|^2}{\|\boldsymbol{\mathscr{M}}\|^2} \leq \, & \frac{64\alpha\left((4+\beta\tau)e^{\beta}-4\right)}{(e^{\beta}-\beta-1)\beta^3\tau}
		\left(\alpha\left(R\sqrt{R}\right)^{N-1}+\alpha\left(\tilde{R}\sqrt{\tilde{R}}\right)^{N-1}+2\right)\frac{\sqrt{\sum_{n=1}^{N}I_n}}{\sqrt{|\Omega|}} 
	\end{align}
	with probability greater than $1-\frac{2}{|\Omega|}$.
\end{theorem}

See Section \ref{ztpSec} for the proof. The result provides an explicit error bound of our methodology with respect to the CP rank and nonnegative CP rank. This statement is more general than what Theorem \ref{simplethm} permits, mainly since we may choose $\tilde{R} < R$, i.e., the rank of the estimate $\widetilde{\boldsymbol{\mathscr{M}}}$ may be smaller than the rank of the tensor of interest $\boldsymbol{\mathscr{M}}$. We stress that such a rank value for which (\ref{feasible0}) holds may not exist since, in general, this assumption is only guaranteed when $\tilde{R} \geq R$, e.g., by setting
\[
\widetilde{\boldsymbol{\mathscr{M}}} = \argmax_{\boldsymbol{\mathscr{T}}\in S_{\tilde{R}}^+(\beta,\alpha)}\tilde{f}_{\Gamma}(\boldsymbol{\mathscr{T}}),
\]
a feasible problem since $\boldsymbol{\mathscr{M}}\in S_{\tilde{R}}^+(\beta,\alpha)$ for $\tilde{R} \geq R$ whose output will satisfy (\ref{feasible0}). 

Despite this, we state Theorem \ref{mainthm} in this flexible manner since a practitioner is typically oblivious to the model's true structure, so $\tilde{R}$ will likely be chosen smaller than $R$ in practice. In such a scenario, the main result remains applicable and informative for practitioners. As a silver lining, tensors suffer from degeneracy \cite{kolda}, i.e., tensors may be approximated arbitrarily well by a factorization of lower rank. It is therefore conceivable that even when the true rank is known there may exist $\tilde{R}<R$ and $\widetilde{\boldsymbol{\mathscr{M}}}$ satisfying (\ref{feasible0}), which will reduce the numerical complexity involved in producing such an estimate.

Next, we present the main result for the oracle estimator. This is the estimator $\widehat{\boldsymbol{\mathscr{M}}}$ that achieves Poisson tensor completion on the set of true counts, introduced in Section \ref{sec:intro} as the ideal method that we compare our proposed approach to. The statement for the oracle scenario is very similar to the zero-truncated case in Theorem \ref{mainthm0}, but does not consider the set of nonzero entries $\Gamma$. Though Theorem 
\ref{mainthm0} is this work's main contribution due to the novel methodology, the following result may be of independent interest to the reader since it generalizes the work in \cite{PoissonMC} to the tensor case with best sampling complexity to date for general arrays with $N\geq 3$.

\begin{theorem}
	\label{mainthm}
	Under the setup of Theorem \ref{mainthm0}, fix $\hat{R}\in\mathbb{N}$. Then for any $\widehat{\boldsymbol{\mathscr{M}}}\in S_{\hat{R}}^+(\beta,\alpha)$ such that
	\begin{equation}
		\label{feasible}
		f_{\Omega}(\widehat{\boldsymbol{\mathscr{M}}}) \geq f_{\Omega}(\boldsymbol{\mathscr{M}}),
	\end{equation}
	we have
	\begin{equation}
		\label{bound1}
		\frac{\|\boldsymbol{\mathscr{M}}-\widehat{\boldsymbol{\mathscr{M}}}\|^2}{\|\boldsymbol{\mathscr{M}}\|^2} \leq \frac{128\alpha}{\beta^3\tau}\left(\frac{(\alpha R+\alpha\hat{R}+2)\sqrt{\sum_{n=1}^{N}I_n}}{\sqrt{|\Omega|}}\right)
	\end{equation}
	with probability exceeding $1-\frac{2}{|\Omega|}$. Furthermore, in the general case where $\boldsymbol{\mathscr{M}}\in S_{R}(\beta,\alpha)$ and $\widehat{\boldsymbol{\mathscr{M}}}\in S_{\hat{R}}(\beta,\alpha)$ but otherwise under the same assumptions, we have
	\begin{align}
		\label{bound2}
		\frac{\|\boldsymbol{\mathscr{M}}-\widehat{\boldsymbol{\mathscr{M}}}\|^2}{\|\boldsymbol{\mathscr{M}}\|^2} \leq \, & \frac{128\alpha}{\beta^3\tau}
		\left(\alpha\left(R\sqrt{R}\right)^{N-1}+\alpha\left(\hat{R}\sqrt{\hat{R}}\right)^{N-1}+2\right)\frac{\sqrt{\sum_{n=1}^{N}I_n}}{\sqrt{|\Omega|}} 
	\end{align}
	with probability greater than $1-\frac{2}{|\Omega|}$.
\end{theorem}
The proof is postponed until Section \ref{PTCresult}. 

In the error bound of the oracle estimator, we see a simplified right hand side (\ref{bound1}) in contrast to the zero-truncated Poisson estimator error bound (\ref{bound10}) which contains the multiplicative term
\[
\kappa = \frac{(4+\beta\tau)e^{\beta}-4}{2(e^{\beta}-\beta-1)}.
\]
This is the error amplification factor defined in (\ref{def-kappa}) that we encounter in the error bound (\ref{errorbounds_intro}) of the introductory result. This observation and considering simplified circumstances provide the proof of Theorem \ref{simplethm}, which is in fact a corollary of Theorems \ref{mainthm0} and \ref{mainthm} that presents the derived error bounds together.

\begin{proof}[Proof of Theorem \ref{simplethm}]
In the setting of Theorem \ref{simplethm}, the conditions of Theorems \ref{mainthm0} and \ref{mainthm} are satisfied with $R = \tilde{R} = \hat{R}$. Applying both of these results, error bounds (\ref{bound10}) and (\ref{bound1}) hold simultaneously with probability exceeding $1-\frac{4}{|\Omega|}$ by a union bound.

Using our previous observations on the term $\kappa$, the following inequalities both hold with probability exceeding $1-\frac{4}{|\Omega|}$
\[
\frac{\|\boldsymbol{\mathscr{M}}-\widetilde{\boldsymbol{\mathscr{M}}}\|^2}{\|\boldsymbol{\mathscr{M}}\|^2} \leq \kappa\cdot\frac{128\alpha}{\beta^3}\left(\alpha(e^2-2) + 3\log_2(|\Omega|) \right)\frac{(2\alpha R+2)\sqrt{\sum_{n=1}^{N}I_n}}{\sqrt{|\Omega|}}
\]
and
\[
\frac{\|\boldsymbol{\mathscr{M}}-\widehat{\boldsymbol{\mathscr{M}}}\|^2}{\|\boldsymbol{\mathscr{M}}\|^2} \leq \frac{128\alpha}{\beta^3}\left(\alpha(e^2-2) + 3\log_2(|\Omega|) \right)\frac{(2\alpha R+2)\sqrt{\sum_{n=1}^{N}I_n}}{\sqrt{|\Omega|}}.
\]

To simplify further, notice that with $I = \max_{n}I_n$ we have $\sqrt{\sum_{n=1}^{N}I_n}\leq \sqrt{NI}$ and $\log_2(|\Omega|)\leq N\log_2(I)$. Defining $\kappa$ as before and
\[
\epsilon := \frac{128\alpha}{\beta^3}\left(\alpha(e^2-2) + 3N\log_2(I) \right)\frac{(2\alpha R+2)\sqrt{NI}}{\sqrt{|\Omega|}},
\]
we obtain error bounds
\[
\frac{\|\boldsymbol{\mathscr{M}}-\widetilde{\boldsymbol{\mathscr{M}}}\|^2}{\|\boldsymbol{\mathscr{M}}\|^2} \leq \kappa\epsilon \ \ \ \ \mbox{and} \ \ \ \ \frac{\|\boldsymbol{\mathscr{M}}-\widehat{\boldsymbol{\mathscr{M}}}\|^2}{\|\boldsymbol{\mathscr{M}}\|^2} \leq \epsilon .
\]
This concludes the proof of Theorem \ref{simplethm}, where the statement treats $N, \alpha, \beta\sim O(1)$ in order to write $\epsilon \sim \mathcal{O}(RI^{\frac12}\log_2(I)|\Omega|^{-\frac12})$ for ease of exposition.
\end{proof}

The main focus of this work deals with the nonnegative CP decomposition and rank of tensors. In terms of the general CP rank, notice that bounds (\ref{bound20}) and (\ref{bound2}) exhibit polynomial dependence $(R\sqrt{R})^{N-1}$ on the rank due to the novel work of \cite{navid1,navid2,navid3}. While pessimistic, the approach improves on all tensor sampling complexity results to date, particularly on the dependence of the ambient dimensions $\sum_nI_n$ (see Section \ref{reqlem} for further discussion). A minor contribution of this work is that the same proof strategy can be applied to the nonnegative CP rank with severely improved rank dependence. 

Sections \ref{ztpSec} and \ref{PTCresult} prove Theorems \ref{mainthm0} and \ref{mainthm} respectively. We note that the proof of both results is very similar, where the proof of Theorem \ref{mainthm0} requires several additional steps. For this reason, we prove the zero-truncated result first which allows an expedited proof of Theorem \ref{mainthm}.

\subsection{Zero-Truncated Poisson Tensor Completion: Proof}
\label{ztpSec}
In this section we prove Theorem \ref{mainthm0}, which derives an error bound for our proposed estimator $\widetilde{\boldsymbol{\mathscr{M}}}$ achieving zero-truncated Poisson tensor completion. The proof requires two lemmas, Lemma \ref{KLbound0} and Lemma \ref{mainlemma0} below, which we only state in this section and prove in Sections \ref{KL0} and \ref{secmainlem0} respectively. 

To briefly summarize the proof and the role of the lemmas, Theorem \ref{mainthm0} controls the error $\|\boldsymbol{\mathscr{M}}-\widetilde{\boldsymbol{\mathscr{M}}}\|$ via the largest deviation of the log-likelihood function from its expected value over all feasible tensors (where the expectation is taken in terms of the randomly generated counts $\boldsymbol{\mathscr{X}}$). Lemma \ref{KLbound0} upper bounds the error between $\boldsymbol{\mathscr{M}}$ and $\widetilde{\boldsymbol{\mathscr{M}}}$ on $\Omega$ via the KL divergence of two zero-truncated Poisson probability distributions. The main bulk in the proof of Theorem \ref{mainthm0} shows that the KL divergence between $\boldsymbol{\mathscr{M}}$ and $\widetilde{\boldsymbol{\mathscr{M}}}$ is in turn controlled by the supremum of $|\tilde{f}_{\Gamma}(\boldsymbol{\mathscr{T}})-\mathbb{E}\tilde{f}_{\Gamma}(\boldsymbol{\mathscr{T}})|$ over all $\boldsymbol{\mathscr{T}}\in S_R^+(\beta,\alpha)$. This latter term is bounded by Lemma \ref{mainlemma0}, which derives an upper bound of the supremum in terms of $\alpha, \beta$, the rank, and the tensor dimensions that holds with high probability. The proof ends by applying the uniform random distribution of $\Omega$ to extend the error to all entries, i.e., $\|\boldsymbol{\mathscr{M}}-\widetilde{\boldsymbol{\mathscr{M}}}\|$.

For the first lemma, we define the KL divergence between two zero-truncated Poisson probability distributions $p,q>0$ as
\begin{equation}
	\label{KLdef0}
	D_0(p\| q) \coloneqq \frac{p}{1-e^{-p}}\log\left(\frac{p}{q}\right) - (\log(e^p-1)-\log(e^q-1)).
\end{equation}
The following result lower bounds this KL divergence by the squared difference of two probability distributions.
\begin{lemma}
	\label{KLbound0}
	For any $p,q\in [\beta,\alpha]$, we have
	\[
	(1-e^{-p})D_0(p\| q) \geq \frac{e^{\beta}-\beta-1}{2\alpha(e^{\beta}-1)}(p-q)^2\geq 0.
	\]
\end{lemma}
This result will be used in the proof of Theorem \ref{mainthm0} to translate an upper bound on the KL divergence to an upper bound on the relative error. We postpone the proof until Section \ref{KL0}, but comment that as a consequence of the proof it can be shown that $D_0(p\| q)=0$ if and only if $p=q$. 

The second lemma is the main component in the proof of Theorem \ref{mainthm0}. The result bounds the largest deviation, over all $\boldsymbol{\mathscr{T}}\in S_R^+(\beta,\alpha)$, of the log-likelihood function from its expected value, where the expectation is taken in terms of the observed $\boldsymbol{\mathscr{X}}$. In the proof of Theorem \ref{mainthm0} this term will be shown to dominate the distance between $\boldsymbol{\mathscr{M}}$ and $\widetilde{\boldsymbol{\mathscr{M}}}$, where this distance will depend on $\alpha, \beta$, the rank, and the tensor dimensions. We note that this result holds for any deterministic set of observed entries, $\Omega$.

\begin{lemma}
	\label{mainlemma0}
	Let $\Omega\subseteq [I_1]\times\cdots\times[I_N]$ be any subset of entries and $\boldsymbol{\mathscr{X}}\in \mathbb{Z}_+^{I_1\times\cdots \times I_N}$ be generated as in Theorem \ref{mainthm0} with $\Gamma\subseteq\Omega$ indicating the set of nonzero entries of $\boldsymbol{\mathscr{X}}$ restricted to $\Omega$. Define
 \[
 \tau \coloneqq \frac{1}{\alpha(e^2-2) + 3\log_2(|\Omega|)}
 \]
 and, given $R\in\mathbb{N}$,
 \[
 R_M^+ \coloneqq \sup_{\boldsymbol{\mathscr{T}}\in S_R^+(\beta,\alpha)}\|\boldsymbol{\mathscr{T}}\|_M.
 \]
 
 Then 
	\begin{align}
		\label{concentration10}
		&\sup_{\boldsymbol{\mathscr{T}}\in S_R^+(\beta,\alpha)}|\tilde{f}_{\Gamma}(\boldsymbol{\mathscr{T}})-\mathbb{E}\tilde{f}_{\Gamma}(\boldsymbol{\mathscr{T}})| \leq 32\left(\frac{(4+\beta\tau)e^{\beta}-4}{(e^{\beta}-1)\beta\tau}\right)(R_M^+ + 1) \sqrt{|\Omega|\sum_{n=1}^{N}I_n}, 
	\end{align}
	with probability exceeding $1-\frac{1}{|\Omega|}$, where the probability and expectation are both over the draw of $\boldsymbol{\mathscr{X}}$. Furthermore, under the same assumptions with $R_M \coloneqq \sup_{\boldsymbol{\mathscr{T}}\in S_R(\beta,\alpha)}\|\boldsymbol{\mathscr{T}}\|_M$ we have
	\begin{align}
		\label{concentration20}
		&\sup_{\boldsymbol{\mathscr{T}}\in S_R(\beta,\alpha)}|\tilde{f}_{\Gamma}(\boldsymbol{\mathscr{T}})-\mathbb{E}\tilde{f}_{\Gamma}(\boldsymbol{\mathscr{T}})| \leq 32\left(\frac{(4+\beta\tau)e^{\beta}-4}{(e^{\beta}-1)\beta\tau}\right)(R_M + 1)\sqrt{|\Omega|\sum_{n=1}^{N}I_n}, 
	\end{align}
	with probability exceeding $1-\frac{1}{|\Omega|}$.
\end{lemma}
The proof can be found in Section \ref{secmainlem0}. We may now proceed to the proof of Theorem \ref{mainthm0}.

\begin{proof}[Proof of Theorem \ref{mainthm0}]
	We will first show (\ref{bound10}). Afterward, establishing bound (\ref{bound20}) only requires a minor modification. 
 
 We begin by computing $\mathbb{E}\tilde{f}_{\Gamma}(\boldsymbol{\mathscr{T}})$ for $\boldsymbol{\mathscr{T}}\in\mathbb{R}_+^{I_1\times\cdots\times I_N}$, where the expectation is taken with respect to $\boldsymbol{\mathscr{X}}$ (recall that $\tilde{f}_{\Gamma}$ depends on $\boldsymbol{\mathscr{X}}$). Let $\boldsymbol{\mathscr{U}}$ be a random binary tensor with entries generated as
	\[
	u_{\textbf{i}} \coloneqq
	\left\{
	\begin{array}{ll}
		0  & \mbox{if } x_{\textbf{i}} = 0 \\
		1 & \mbox{if } x_{\textbf{i}} \neq 0,
	\end{array}
	\right.
	\]
	which allows us to write
	\begin{align*}
		\tilde{f}_{\Gamma}(\boldsymbol{\mathscr{T}}) &= \sum_{\textbf{i}\in\Omega}u_{\textbf{i}}\Big[x_{\textbf{i}}\log\left(t_{\textbf{i}}\right)-\log\left(\exp\left(t_{\textbf{i}}\right)-1\right)\Big] \\
		&= \sum_{\textbf{i}\in\Omega}x_{\textbf{i}}\log\left(t_{\textbf{i}}\right)-u_{\textbf{i}}\log\left(\exp\left(t_{\textbf{i}}\right)-1\right).
	\end{align*}
	Notice that $\mathbb{E}u_{\textbf{i}} = 1-\mathbb{P}(x_{\textbf{i}}=0) = 1-\exp\left(-m_{\textbf{i}}\right)$, and therefore
	\[
	\mathbb{E}\tilde{f}_{\Gamma}(\boldsymbol{\mathscr{T}}) = \sum_{\textbf{i}\in\Omega}m_{\textbf{i}}\log\left(t_{\textbf{i}}\right)-\left(1-\exp\left(-m_{\textbf{i}}\right)\right)\log\left(\exp\left(t_{\textbf{i}}\right)-1\right).
	\]

	With this in mind, we now show that the KL divergence of $\boldsymbol{\mathscr{M}}$ and $\widetilde{\boldsymbol{\mathscr{M}}}$ is bounded by the supremum of $|\tilde{f}_{\Gamma}(\boldsymbol{\mathscr{T}})-\mathbb{E}\tilde{f}_{\Gamma}(\boldsymbol{\mathscr{T}})|$ over all $\boldsymbol{\mathscr{T}}\in S_R^+(\beta,\alpha)$. Apply our assumptions on $\boldsymbol{\mathscr{M}}\in S_R^+(\beta,\alpha)$ and $\widetilde{\boldsymbol{\mathscr{M}}}\in S_{\tilde{R}}^+(\beta,\alpha)$ and insert terms that take the marginal expectation with respect to $\boldsymbol{\mathscr{X}}$ only to obtain
	\begin{align*}
		0 &\leq \tilde{f}_{\Gamma}(\widetilde{\boldsymbol{\mathscr{M}}})-\tilde{f}_{\Gamma}(\boldsymbol{\mathscr{M}})\\
		&= \mathbb{E}\Big[\tilde{f}_{\Gamma}(\widetilde{\boldsymbol{\mathscr{M}}})-\tilde{f}_{\Gamma}(\boldsymbol{\mathscr{M}})\Big] + \left(\tilde{f}_{\Gamma}(\widetilde{\boldsymbol{\mathscr{M}}})- \mathbb{E}\tilde{f}_{\Gamma}(\widetilde{\boldsymbol{\mathscr{M}}})\right)+\left(\mathbb{E}\tilde{f}_{\Gamma}(\boldsymbol{\mathscr{M}})-\tilde{f}_{\Gamma}(\boldsymbol{\mathscr{M}})\right)\\
		&\leq \mathbb{E}\Big[\tilde{f}_{\Gamma}(\widetilde{\boldsymbol{\mathscr{M}}})-\tilde{f}_{\Gamma}(\boldsymbol{\mathscr{M}})\Big] + \sup_{\boldsymbol{\mathscr{T}}\in S_{\tilde{R}}^+(\beta,\alpha)}\Bigg|\tilde{f}_{\Gamma}(\boldsymbol{\mathscr{T}})- \mathbb{E}\tilde{f}_{\Gamma}(\boldsymbol{\mathscr{T}})\Bigg| + \sup_{\boldsymbol{\mathscr{T}}\in S_R^+(\beta,\alpha)}\Bigg|\tilde{f}_{\Gamma}(\boldsymbol{\mathscr{T}})- \mathbb{E}\tilde{f}_{\Gamma}(\boldsymbol{\mathscr{T}})\Bigg|\\
		&= -\sum_{\textbf{i}\in\Omega}\Bigg[m_{\textbf{i}}\log\left(\frac{m_{\textbf{i}}}{\tilde{m}_{\textbf{i}}}\right)-\left(1-\exp\left(-m_{\textbf{i}}\right)\right)\left(\log\left(\exp\left(m_{\textbf{i}}\right)-1\right)-\log\left(\exp\left(\tilde{m}_{\textbf{i}}\right)-1\right)\right)\Bigg] \\ 
		&\phantom{\qquad} + \sup_{\boldsymbol{\mathscr{T}}\in S_{\tilde{R}}^+(\beta,\alpha)}\Bigg|\tilde{f}_{\Gamma}(\boldsymbol{\mathscr{T}})- \mathbb{E}\tilde{f}_{\Gamma}(\boldsymbol{\mathscr{T}})\Bigg| + \sup_{\boldsymbol{\mathscr{T}}\in S_R^+(\beta,\alpha)}\Bigg|\tilde{f}_{\Gamma}(\boldsymbol{\mathscr{T}})- \mathbb{E}\tilde{f}_{\Gamma}(\boldsymbol{\mathscr{T}})\Bigg|\\
		&= -\sum_{\textbf{i}\in\Omega}\left(1-\exp\left(-m_{\textbf{i}}\right)\right)D_0\left(m_{\textbf{i}}\|\tilde{m}_{\textbf{i}}\right) \\ 
		&\phantom{\qquad} + \sup_{\boldsymbol{\mathscr{T}}\in S_{\tilde{R}}^+(\beta,\alpha)}\Bigg|\tilde{f}_{\Gamma}(\boldsymbol{\mathscr{T}})- \mathbb{E}\tilde{f}_{\Gamma}(\boldsymbol{\mathscr{T}})\Bigg| + \sup_{\boldsymbol{\mathscr{T}}\in S_R^+(\beta,\alpha)}\Bigg|\tilde{f}_{\Gamma}(\boldsymbol{\mathscr{T}})- \mathbb{E}\tilde{f}_{\Gamma}(\boldsymbol{\mathscr{T}})\Bigg|.
	\end{align*}
	In the last line we used the definition of the KL divergence between two zero-truncated Poisson probability distributions (\ref{KLdef0}). 
 
 Since $m_{\textbf{i}},\tilde{m}_{\textbf{i}}\in[\beta,\alpha]$ for all $\textbf{i}\in [I_1]\times\cdots\times [I_N]$, using Lemma \ref{KLbound0}, this term can be lower bounded as
	\[
	\sum_{\textbf{i}\in\Omega}\left(1-\exp\left(-m_{\textbf{i}}\right)\right)D_0\left(m_{\textbf{i}}\|\tilde{m}_{\textbf{i}}\right) \geq \frac{e^{\beta}-\beta-1}{2\alpha(e^{\beta}-1)}\sum_{\textbf{i}\in\Omega}\left(m_{\textbf{i}}-\tilde{m}_{\textbf{i}}\right)^2,
	\]
which translates our bound on the KL divergence to the usual Euclidean distance between $\boldsymbol{\mathscr{M}}$ and $\widetilde{\boldsymbol{\mathscr{M}}}$ (on $\Omega$). Gathering our bounds and applying equation (\ref{concentration10}) from Lemma \ref{mainlemma0} for both $R$ and $\tilde{R}$, we have established that for any $\Omega$
	\begin{align}
		\label{eq10}
		&\frac{e^{\beta}-\beta-1}{2\alpha(e^{\beta}-1)}\sum_{\textbf{i}\in\Omega}\left(m_{\textbf{i}}-\tilde{m}_{\textbf{i}}\right)^2\\ 
		&\leq \sup_{\boldsymbol{\mathscr{T}}\in S_{\tilde{R}}^+(\beta,\alpha)}\Bigg|\tilde{f}_{\Gamma}(\boldsymbol{\mathscr{T}})- \mathbb{E}\tilde{f}_{\Gamma}(\boldsymbol{\mathscr{T}})\Bigg| + \sup_{\boldsymbol{\mathscr{T}}\in S_R^+(\beta,\alpha)}\Bigg|\tilde{f}_{\Gamma}(\boldsymbol{\mathscr{T}})- \mathbb{E}\tilde{f}_{\Gamma}(\boldsymbol{\mathscr{T}})\Bigg|\nonumber \\
		&\leq 32\left(\frac{(4+\beta\tau)e^{\beta}-4}{(e^{\beta}-1)\beta\tau}\right)(R_{M}^{+}+\tilde{R}_{M}^{+}+2)\sqrt{|\Omega|\sum_{n=1}^{N}I_n}, \nonumber
	\end{align}
	with probability exceeding $1-\frac{2}{|\Omega|}$ by a union bound, where $R_{M}^{+}$ and $\tilde{R}_{M}^{+}$ are defined as in Lemma \ref{mainlemma0}. 
 
 We now apply our uniform random assumption on $\Omega$ to extend the error above to all entries (i.e., not just in $\Omega$). Notice that in terms of the distribution on $\Omega$, the final term above is deterministic since its cardinality $|\Omega|$ is fixed for all outcomes. Therefore, given $\boldsymbol{\mathscr{X}}$ such that the bound holds, we have bounded the random variable $\sum_{\textbf{i}\in\Omega}\left(m_{\textbf{i}}-\tilde{m}_{\textbf{i}}\right)^2$. Since $\boldsymbol{\mathscr{X}}$ and $\Omega$ are independently generated, the upper bound holds for the expected value over $\Omega$ as well, i.e.,
	\begin{align*}
		&\mathbb{E}\sum_{\textbf{i}\in\Omega}\left(m_{\textbf{i}}-\tilde{m}_{\textbf{i}}\right)^2 \leq 64\left(\frac{(4+\beta\tau)e^{\beta}-4}{(e^{\beta}-\beta-1)\beta\tau}\right)\alpha(R_{M}^{+}+\tilde{R}_{M}^{+}+2)\sqrt{|\Omega|\sum_{n=1}^{N}I_n}.
	\end{align*}
	We finish the proof by computing the expected value above. Define 
	$K \coloneqq \binom{I_1I_2\cdots I_N}{|\Omega|}$, 
	which is the number of subsets of $[I_1]\times\cdots\times [I_N]$ of size $|\Omega|$ and let $\{\Omega_{k}\}_{k=1}^{K}$ list all such subsets. Then
	\begin{align*}
		&\mathbb{E}\sum_{\textbf{i}\in\Omega}\left(m_{\textbf{i}}-\tilde{m}_{\textbf{i}}\right)^2 = \frac{1}{K}\sum_{k=1}^{K}\sum_{\textbf{i}\in\Omega_{k}}\left(m_{\textbf{i}}-\tilde{m}_{\textbf{i}}\right)^2\\ 
		&= \frac{1}{K}\sum_{\textbf{i}\in[I_1]\times\cdots\times [I_N]}\binom{I_1\cdots I_N-1}{|\Omega|-1}\left(m_{\textbf{i}}-\tilde{m}_{\textbf{i}}\right)^2,
	\end{align*}
	where the last equality holds since for any tensor entry $\textbf{i}\in [I_1]\times\cdots\times [I_N]$ there will be a total of $\binom{I_1\cdots I_N-1}{|\Omega|-1}$ subsets of size $|\Omega|$ that contain $\textbf{i}$. Therefore, in the sum over $k$ each term $\left(m_{\textbf{i}}-\tilde{m}_{\textbf{i}}\right)^2$ will appear exactly $\binom{I_1\cdots I_N-1}{|\Omega|-1}$ times. The proof ends by noticing that
	\[
	\frac{1}{K}\binom{I_1\cdots I_N-1}{|\Omega|-1} = \binom{I_1\cdots I_N}{|\Omega|}^{-1}\binom{I_1\cdots I_N-1}{|\Omega|-1} = \frac{|\Omega|}{I_1\cdots I_N}
	\]
	and
	\[
	\frac{|\Omega|}{I_1\cdots I_N}\sum_{\textbf{i}\in[I_1]\times\cdots\times [I_N]}\left(m_{\textbf{i}}-\tilde{m}_{\textbf{i}}\right)^2 = \frac{|\Omega|\|\boldsymbol{\mathscr{M}}-\widetilde{\boldsymbol{\mathscr{M}}}\|^2}{I_1\cdots I_N} \geq \frac{|\Omega|\beta^2\|\boldsymbol{\mathscr{M}}-\widetilde{\boldsymbol{\mathscr{M}}}\|^2}{\|\boldsymbol{\mathscr{M}}\|^2}.
	\]
Finally, by equation \eqref{r+} in Lemma \ref{lemma3}, we have $R_{M}^{+}\leq\alpha R$ and $\tilde{R}_{M}^{+}\leq\alpha \tilde{R}$ which finishes the proof.
 
	The proof of (\ref{bound20}) is analogous with respect to $S_R(\beta,\alpha)$ and $S_{\tilde{R}}(\beta,\alpha)$, where $R_{M}^{+}$ and $\tilde{R}_{M}^{+}$ are replaced with $R_{M}$ and $\tilde{R}_{M}$ respectively in the proof above. This replaces the term $\alpha\tilde{R}+\alpha R$ in (\ref{bound10}) with $\alpha(\tilde{R}\sqrt{\tilde{R}})^{N-1} + \alpha(R\sqrt{R})^{N-1}$ using equation \eqref{r} in Lemma \ref{lemma3}. The remaining terms are unchanged and the result follows.
\end{proof}

\subsection{Poisson Tensor Completion Proof}
\label{PTCresult}

The proof of Theorem \ref{mainthm} is very similar to the proof of Theorem \ref{mainthm0}. For brevity, we will refer the reader to the proof of Theorem \ref{mainthm0} when similar steps are applied. The main difference will be to consider instead the KL divergence between Poisson probability distributions, defined as
\begin{equation}
	\label{KLdef}
	D(p\| q) \coloneqq p\log\left(\frac{p}{q}\right) - (p-q).
\end{equation}
The first lemma establishes a lower bound for the KL divergence.
\begin{lemma}
	\label{KLbound}
	For any $p,q\in (0,\alpha]$, we have
	\[
	D(p\| q) \geq \frac{(p-q)^2}{2\alpha}.
	\]
\end{lemma}
The proof of this lemma is postponed until Section \ref{KLlemma}. The second lemma is an analogous version of Lemma \ref{mainlemma0} used for the zero-truncated result.

\begin{lemma}
	\label{mainlemma}
	Let $\Omega\subseteq [I_1]\times\cdots\times[I_N]$ be any subset of entries, $\boldsymbol{\mathscr{X}}\in \mathbb{Z}_+^{I_1\times\cdots \times I_N}$ be generated as in Theorem \ref{mainthm}, and the function $f_{\Omega}$ (which depends on $\boldsymbol{\mathscr{X}}$) be defined as in (\ref{likelihood}). Define
 \[
 \tau \coloneqq \frac{1}{\alpha(e^2-2) + 3\log_2(|\Omega|)}
 \]
 and, given $R\in\mathbb{N}$,
 \[
 R_M^+ \coloneqq \sup_{\boldsymbol{\mathscr{T}}\in S_R^+(\beta,\alpha)}\|\boldsymbol{\mathscr{T}}\|_M.
 \]
 
 Then
	\begin{align}
		\label{concentration1}
		\sup_{\boldsymbol{\mathscr{T}}\in S_R^+(\beta,\alpha)}|f_{\Omega}(\boldsymbol{\mathscr{T}})-\mathbb{E}f_{\Omega}(\boldsymbol{\mathscr{T}})| \leq \frac{64(R_{M}^{+}+1)}{\beta\tau}\sqrt{|\Omega|\sum_{n=1}^{N}I_n},
	\end{align}
	with probability exceeding $1-\frac{1}{|\Omega|}$, where the probability and expectation are both over the draw of $\boldsymbol{\mathscr{X}}$. Furthermore, under the same assumptions with $R_M \coloneqq \sup_{\boldsymbol{\mathscr{T}}\in S_R(\beta,\alpha)}\|\boldsymbol{\mathscr{T}}\|_M$ we have
	\begin{align}
		\label{concentration2}
		\sup_{\boldsymbol{\mathscr{T}}\in S_R(\beta,\alpha)}|f_{\Omega}(\boldsymbol{\mathscr{T}})-\mathbb{E}f_{\Omega}(\boldsymbol{\mathscr{T}})| \leq \frac{64(R_M+1)}{\beta\tau}\sqrt{|\Omega|\sum_{n=1}^{N}I_n},
	\end{align}
	with probability exceeding $1-\frac{1}{|\Omega|}$.
\end{lemma}

See Section \ref{secmainlem} for the proof. We may now proceed to the proof of Theorem \ref{mainthm}.

\begin{proof}[Proof of Theorem \ref{mainthm}]
	We will first show (\ref{bound1}). Afterward, establishing (\ref{bound2}) only requires a minor modification. We begin by noting that for any $\boldsymbol{\mathscr{T}}\in\mathbb{R}_+^{I_1\times\cdots\times I_N}$
	\[
	\mathbb{E}f_{\Omega}(\boldsymbol{\mathscr{T}}) = \mathbb{E}\sum_{\textbf{i}\in\Omega}x_{\textbf{i}}\log\left(t_{\textbf{i}}\right)-t_{\textbf{i}} = \sum_{\textbf{i}\in\Omega}m_{\textbf{i}}\log\left(t_{\textbf{i}}\right)-t_{\textbf{i}},
	\]
	where the expectation is taken with respect to $\boldsymbol{\mathscr{X}}$. Applying our assumptions on $\boldsymbol{\mathscr{M}}\in S_R^+(\beta,\alpha)$ and $\widehat{\boldsymbol{\mathscr{M}}}\in S_{\hat{R}}^+(\beta,\alpha)$, we insert terms that take the marginal expectation with respect to $\boldsymbol{\mathscr{X}}$ only and obtain
	\begin{align*}
		0 &\leq f_{\Omega}(\widehat{\boldsymbol{\mathscr{M}}})-f_{\Omega}(\boldsymbol{\mathscr{M}})
		= \mathbb{E}\Big[f_{\Omega}(\widehat{\boldsymbol{\mathscr{M}}})-f_{\Omega}(\boldsymbol{\mathscr{M}})\Big] + \left(f_{\Omega}(\widehat{\boldsymbol{\mathscr{M}}})- \mathbb{E}f_{\Omega}(\widehat{\boldsymbol{\mathscr{M}}})\right)+\left(\mathbb{E}f_{\Omega}(\boldsymbol{\mathscr{M}})-f_{\Omega}(\boldsymbol{\mathscr{M}})\right)\\
		&\leq \mathbb{E}\Big[f_{\Omega}(\widehat{\boldsymbol{\mathscr{M}}})-f_{\Omega}(\boldsymbol{\mathscr{M}})\Big] + \sup_{\boldsymbol{\mathscr{T}}\in S_{\hat{R}}^+(\beta,\alpha)}\Bigg|f_{\Omega}(\boldsymbol{\mathscr{T}})- \mathbb{E}f_{\Omega}(\boldsymbol{\mathscr{T}})\Bigg| + \sup_{\boldsymbol{\mathscr{T}}\in S_R^+(\beta,\alpha)}\Bigg|f_{\Omega}(\boldsymbol{\mathscr{T}})- \mathbb{E}f_{\Omega}(\boldsymbol{\mathscr{T}})\Bigg|\\
		&= -\sum_{\textbf{i}\in\Omega}\Bigg[m_{\textbf{i}}\log\left(\frac{m_{\textbf{i}}}{\hat{m}_{\textbf{i}}}\right)-\left(m_{\textbf{i}}-\hat{m}_{\textbf{i}}\right)\Bigg] \\
		&\phantom{\qquad} + \sup_{\boldsymbol{\mathscr{T}}\in S_{\hat{R}}^+(\beta,\alpha)}\Bigg|f_{\Omega}(\boldsymbol{\mathscr{T}})- \mathbb{E}f_{\Omega}(\boldsymbol{\mathscr{T}})\Bigg| + \sup_{\boldsymbol{\mathscr{T}}\in S_R^+(\beta,\alpha)}\Bigg|f_{\Omega}(\boldsymbol{\mathscr{T}})- \mathbb{E}f_{\Omega}(\boldsymbol{\mathscr{T}})\Bigg|\\
		&= -\sum_{\textbf{i}\in\Omega}D\left(m_{\textbf{i}}\|\hat{m}_{\textbf{i}}\right) + \sup_{\boldsymbol{\mathscr{T}}\in S_{\hat{R}}^+(\beta,\alpha)}\Bigg|f_{\Omega}(\boldsymbol{\mathscr{T}})- \mathbb{E}f_{\Omega}(\boldsymbol{\mathscr{T}})\Bigg| + \sup_{\boldsymbol{\mathscr{T}}\in S_R^+(\beta,\alpha)}\Bigg|F_{\Omega}(\boldsymbol{\mathscr{T}})- \mathbb{E}f_{\Omega}(\boldsymbol{\mathscr{T}})\Bigg|.
	\end{align*}
	In the last line we used the definition of the KL divergence between two Poisson probability distributions (\ref{KLdef}). Since $m_{\textbf{i}},\hat{m}_{\textbf{i}}\in[\beta,\alpha]$ for all $\textbf{i}\in [I_1]\times\cdots\times [I_N]$, using Lemma \ref{KLbound}, this term can be lower bounded as
	\[
	\sum_{\textbf{i}\in\Omega}D\left(m_{\textbf{i}}\|\hat{m}_{\textbf{i}}\right) \geq \frac{1}{2\alpha}\sum_{\textbf{i}\in\Omega}\left(m_{\textbf{i}}-\hat{m}_{\textbf{i}}\right)^2.
	\]
	Gathering our bounds and applying equation (\ref{concentration1}) from Lemma \ref{mainlemma} for both $R$ and $\hat{R}$, we have established that for any $\Omega$
	\begin{align}
		\label{eq1}
		\frac{1}{2\alpha}\sum_{\textbf{i}\in\Omega}\left(m_{\textbf{i}}-\hat{m}_{\textbf{i}}\right)^2
		&\leq \sup_{\boldsymbol{\mathscr{T}}\in S_{\hat{R}}^+(\beta,\alpha)}\Bigg|f_{\Omega}(\boldsymbol{\mathscr{T}})- \mathbb{E}f_{\Omega}(\boldsymbol{\mathscr{T}})\Bigg| + \sup_{\boldsymbol{\mathscr{T}}\in S_R^+(\beta,\alpha)}\Bigg|f_{\Omega}(\boldsymbol{\mathscr{T}})- \mathbb{E}f_{\Omega}(\boldsymbol{\mathscr{T}})\Bigg| \\
		&\leq \frac{64(R_{M}^{+}+\hat{R}_{M}^{+}+2)}{\beta\tau}\sqrt{|\Omega|\sum_{n=1}^{N}I_n},\nonumber
	\end{align}
	with probability exceeding $1-\frac{2}{|\Omega|}$ by a union bound, where $R_{M}^{+}$ and $\hat{R}_{M}^{+}$ are defined as in Lemma \ref{mainlemma}. We now apply our assumption on $\Omega$.
	
	Notice that in terms of the distribution on $\Omega$, the final term above is deterministic since the cardinality $|\Omega|$ is fixed for all outcomes. Therefore, given $\boldsymbol{\mathscr{X}}$ such that the bound holds, we have bounded the random variable $\sum_{\textbf{i}\in\Omega}\left(m_{\textbf{i}}-\hat{m}_{\textbf{i}}\right)^2$. Since $\boldsymbol{\mathscr{X}}$ and $\Omega$ are independently generated, the upper bound holds for the expected value over $\Omega$ as well, i.e.,
	\begin{align*}
		\mathbb{E}\sum_{\textbf{i}\in\Omega}\left(m_{\textbf{i}}-\hat{m}_{\textbf{i}}\right)^2 &= \frac{|\Omega|\|\boldsymbol{\mathscr{M}}-\widehat{\boldsymbol{\mathscr{M}}}\|^2}{I_1\cdots I_N}\leq \frac{128\alpha(R_{M}^{+}+\hat{R}_{M}^{+}+2)}{\beta\tau}\sqrt{|\Omega|\sum_{n=1}^{N}I_n}.
	\end{align*}
	The proof ends by noting that
	\[
	\frac{|\Omega|\|\boldsymbol{\mathscr{M}}-\widehat{\boldsymbol{\mathscr{M}}}\|}{I_1\cdots I_N} \geq \frac{|\Omega|\beta^2\|\boldsymbol{\mathscr{M}}-\widehat{\boldsymbol{\mathscr{M}}}\|^2}{\|\boldsymbol{\mathscr{M}}\|^2}
	\]
 and using Lemma \ref{lemma3} to bound $R_{M}^{+}\leq\alpha R$ and $\hat{R}_{M}^{+}\leq\alpha \hat{R}$.
 
	The proof of (\ref{bound2}) is analogous with respect to $S_R(\beta,\alpha)$ and $S_{\hat{R}}(\beta,\alpha)$, where we replace $R_{M}^{+}$ and $\hat{R}_{M}^{+}$ with $R_{M}$ and $\hat{R}_{M}$ in the proof above. This replaces the term $\alpha\hat{R}+\alpha R$ in (\ref{bound1}) with $\alpha(\hat{R}\sqrt{\hat{R}})^{N-1} + \alpha(R\sqrt{R})^{N-1}$. The remaining terms are unchanged and the result follows.
\end{proof}

%% file: conc.tex
We proposed a novel statistical inference method for zero-congested multiway count data that does not require the user to distinguish between true and false zero counts. This work debuts the approach on the multi-parameter Poisson model, where we condition this distribution on the positive integers in order to appropriately ignore zero values and treat the respective array entries as unobserved. Under a low-rank parametric model, our approach applies zero-truncated Poisson regression only on the non-zeros. The low-dimensional parametric assumption allows us to achieve Poisson estimation on the entire volume in an underdetermined setting that only considers true counts. We show that the approach is efficient at approximating the mean values when the level of zero-inflation is not excessive relative to the parametric complexity. For an $N$-way parametric tensor $\boldsymbol{\mathscr{M}}\in\mathbb{R}^{I\times \cdots\times I}$ with nonnegative CP rank $R$ that generates Poisson observations, our main result states that $\sim IR^2\log_2^2(I)$ non-zeros provide an accurate estimate via our methodology. 

Our numerical experiments explore the implementation of the approach via maximum likelihood and its effectiveness by comparing it to ideal ``oracle'' scenario, in which the locations of false zeros are known. The presented cases show that in many situations our approach is comparable to the oracle while allowing for practical implementation. We explore via numerical experiments the limitations of the method, including its sensitivity to the bounds $\beta$ and $\alpha$ on the Poisson parameters. The experiments reveal that when $\beta$ is not small, say  $\beta \geq .1$, zero truncating the Poisson distribution is an excellent approximation. On the other hand, when the parametric values are small (e.g., $\beta\leq .01$ and $\alpha\leq 1$), the efficiency of our approach is degraded since such situations with sparse data generate an overwhelming amount of true zeros that are neglected.

Several extensions remain to be explored as future work. The current work focuses on the multi-parameter Poisson distribution. However, the paradigm can be applied to any count data model, such as the negative binomial distribution, or even continuous counterparts for other applications (e.g., the normal distribution). Furthermore, we only consider the case of congestion by false zeros since it is the most common type of corruption in the literature of count data. As an extension, any range of integers can be truncated to allow for other types of untrusted count values in data. In the case of continuous models, distributions can be conditioned to any interval of trusted observations. These types of generalizations, paired with more ample theoretical results, can help launch our proposed statistical inference paradigm to handle severe corruption in a wide range of applications that involve multi-dimensional data processing.

%% file: mainlemmas.tex
This appendix is dedicated to the proofs of the lemmas required for the results of Section \ref{sec:proofs}. Section \ref{KLdiv} focuses on the lower bounds for the KL-divergences, while Section \ref{secmainlem} proves the main lemmas to establish our results. Finally Section \ref{proofreqlem} proves additional lemmas required for the proofs in Section \ref{secmainlem}.

\subsection{Lower Bounds for KL Divergence}
\label{KLdiv}

This section proves Lemmas \ref{KLbound} and \ref{KLbound0}. We will first produce the lower bound for the KL-divergence between two Poisson probability distributions, this in turn will be used to obtain the lower bound for the divergence between two zero-truncated Poisson distributions.

\subsubsection{Proof of Lemma \ref{KLbound}}
\label{KLlemma}

Using the work in \cite{prob}, the authors in \cite{PoissonMC} produce a lower bound for the KL divergence between two Poisson probability distributions. In this work, using the work in \cite{prob} we are able to obtain a tighter bound.

\begin{proof}[Proof of Lemma \ref{KLbound}]
	In \cite{prob}, the author establishes in equation 11 of Chapter 3 that
	\[
	(1+x)\log(1+x) = x + \frac{x^2}{2(1+x^*)}
	\]
	holds for $x>-1$ and some $x^*$ between 0 and $x$. With the choice $x = (p-q)/q > -1$, if we multiply through by $q$ we obtain
	\[
	p\log\left(\frac{p}{q}\right)-(p-q) = \frac{(p-q)^2}{2q(1+x^*)}.
	\]
	
	We now lower bound the right hand side by upper bounding the term $1+x^*$, which we note is always strictly positive. Consider the two possible cases $p\geq q$ and $p< q$. When $p\geq q$, we have $x\geq 0$ so that $x^*\in [0,(p-q)/q]$ and therefore
	\[
	1+x^*\leq 1+\frac{p-q}{q}.
	\]
	Otherwise, if $p< q$ then $x^*\in [(p-q)/q,0)$ and
	\[
	1+x^*<1.
	\]
	Using both of these upper bounds, our assumption $p,q\leq \alpha$ gives that
	\[
	\frac{1}{q(1+x^*)} \geq \frac{1}{q}\min\Bigg\{1,\frac{1}{1+\frac{p-q}{q}}\Bigg\} = \min\Bigg\{\frac{1}{q},\frac{1}{p}\Bigg\} \geq \frac{1}{\alpha}
	\]
	and therefore
	\[
	\frac{(p-q)^2}{2q(1+x^*)} \geq \frac{(p-q)^2}{2\alpha}.
	\]
	In terms of the KL divergence between two Poisson probability distributions (\ref{KLdef}), we have shown that for $p,q\in (0,\alpha]$
	\[
	D\left(p\|q\right) \geq \frac{(p-q)^2}{2\alpha}.
	\]
\end{proof}

\subsubsection{Proof of Lemma \ref{KLbound0}}
\label{KL0}
We now prove Lemma \ref{KLbound0}, which applies the lower bound established in Lemma \ref{KLbound}.

\begin{proof}[Proof of Lemma \ref{KLbound0}]
	Using basic calculus, we will show that for some term $c_{\beta}>0$ depending only on $\beta$, we have 
	\[
	(1-e^{-p})D_0(p\|q)\geq c_{\beta}D(p\|q)
	\]
	for all $p,q\geq\beta>0$ where $D(p\|q)$ is defined in (\ref{KLdef}). Using Lemma \ref{KLbound} will then establish the claim.
	
	To this end, let $c_{\beta}>0$ be an arbitrary constant (independent of $p$ and $q$) and consider $p\geq\beta$ fixed, so that we only vary $q$ in $(1-e^{-p})D_0(p\|q)$ and $c_{\beta}D(p\|q)$. Notice that these univariate functions intersect at $q=p$ since $D_0(p\|p)=0=D(p\|p)$. We compute $c_{\beta}$ so that $(1-e^{-p})D_0(p\|q)$ has a greater rate of change than $c_{\beta}D(p\|q)$ for $q> p$. Taking partial derivatives we obtain
	\[
	\partial q \Big[(1-e^{-p})D_0(p\|q)\Big] = \frac{e^q(e^p - 1)}{e^p(e^q - 1)} - \frac{p}{q}
	\]
	and
	\[
	\partial q \Big[c_{\beta}D(p\|q)\Big] = c_{\beta}\left(1-\frac{p}{q}\right).
	\]
	Notice that for $q> p$ we have $\partial q (1-e^{-p})D_0(p\|q)>0$ and $(1-p/q)>0$, and we therefore achieve our greater rate of change if
	\[
	c_{\beta} \leq \frac{\frac{e^q(e^p - 1)}{e^p(e^q - 1)} - \frac{p}{q}}{\left(1-\frac{p}{q}\right)} = \frac{qe^q(e^{p}-1)}{(q-p)e^p(e^q - 1)} - \frac{p}{q-p} \coloneqq f(q)
	\]
	holds for all $p\geq\beta$ and $q>p$.
	
	Examining $f(q)$, we see that $f^{\prime}(q)>0$ for all $q>p$ and therefore $f(q)\geq f(p)$ where
	\[
	f(p) = \lim_{q\rightarrow p}\left(\frac{qe^q(e^{p}-1)}{(q-p)e^p(e^q - 1)} - \frac{p}{q-p}\right) = \frac{e^p-p-1}{e^p-1}.
	\]
	This allows us to choose
	\[
	c_{\beta}\coloneqq \frac{e^{\beta}-\beta-1}{e^{\beta}-1} \leq \frac{e^p-p-1}{e^p-1},
	\]
	where the inequality holds for all $p\geq\beta$ since $f(p)$ is a monotonically increasing function with respect to $p$.
	
	We have chosen $c_{\beta}>0$ such that $(1-e^{-p})D_0(p\|q)$ and $c_{\beta}D(p\|q)$ agree at $q=p$ and $\partial_q(1-e^{-p})D_0(p\|q)\geq\partial_q c_{\beta}D(p\|q)$ when $q>p$. Therefore $(1-e^{-p})D_0(p\|q)\geq c_{\beta}D(p\|q)$ when $q>p$. The same argument can be applied when $q< p$ (but now with negative rates of change), where the same choice for $c_{\beta}$ will give $(1-e^{-p})D_0(p\|q)\geq c_{\beta}D(p\|q)$ when $p>q$. Using Lemma \ref{KLbound}, we have shown for all $p,q\in[\beta,\alpha]$
	\[
	(1-e^{-p})D_0(p\|q)\geq c_{\beta}D(p\|q)\geq \frac{c_{\beta}(p-q)^2}{2\alpha}.
	\]
\end{proof}

\subsection{Proof of the Main Lemmas}
\label{secmainlem}
The main bulk of our work will be to prove Lemmas \ref{mainlemma0} and \ref{mainlemma}, the main components in the proofs of Theorems \ref{mainthm0} and \ref{mainthm}. We note that both proofs are very similar, requiring only different terms but applying the same proof strategy. The proof of Lemma \ref{mainlemma0} requires more terms to be bounded, aside from analogous terms found in the proof of Lemma \ref{mainlemma}. For this reason we will focus on a detailed proof of Lemma \ref{mainlemma0} and as a consequence the proof of Lemma \ref{mainlemma} can be achieved in a condensed manner.

To this end, we collect several additional lemmas that will be used in both proofs.

\subsubsection{Required Lemmas}
\label{reqlem}

We begin by gathering some standard tools from probability in Banach spaces \cite{ledoux}. The following is the symmetrization inequality in diluted form, simplified to be directly applicable to our context (see \cite{ledoux} for the full result).

\begin{lemma}[Symmetrization Inequality, Lemma 6.3 in \cite{ledoux}] 
	\label{lemma2a}
	Let $F:\mathbb{R}_+\mapsto\mathbb{R}_+$ be convex. Let $\{y_{\ell}\}_{\ell=1}^{L}\subset \mathbb{R}$ be a finite sequence of independent random variables with $\mathbb{E}|y_{\ell}|<\infty$ and $\epsilon_1, \epsilon_2, \cdots, \epsilon_L$ be i.i.d. Rademacher random variables. Then for any bounded $U\subset \mathbb{R}$
	\[
	\mathbb{E}F\left(\sup_{(u_1,\cdots,u_L)\in U^L}\Bigg|\sum_{\ell=1}^{L}u_{\ell}(y_{\ell}-\mathbb{E}y_{\ell})\Bigg|\right) \leq \mathbb{E}F\left(2\sup_{(u_1,\cdots,u_L)\in U^L}\Bigg|\sum_{\ell=1}^{L}\epsilon_{\ell}u_{\ell}y_{\ell}\Bigg|\right),
	\]
	where the expected value on the right hand side is taken over $y_{\ell}$ and $\epsilon_{\ell}$.
\end{lemma}
The symmetrization technique is by now standard, allowing simplified computations by translating these with respect to well studied Rademacher random variables. Subsequently, introducing a Rademacher sequence will pair well with the next result.

\begin{lemma}[Contraction Inequality, Theorem 4.12 in \cite{ledoux}] 
	\label{lemma2}
	Let $F:\mathbb{R}_+\mapsto\mathbb{R}_+$ be convex and increasing. For $\ell\in [L]$, let $\epsilon_{\ell}$ be i.i.d. Rademacher random variables and $\varphi_{\ell}:\mathbb{R}\mapsto\mathbb{R}$ be contractions such that $\varphi_{\ell}(0)=0$. Then for any bounded $U\subset\mathbb{R}^{L}$
	\[
	\mathbb{E}F\left(\frac{1}{2}\sup_{(u_1,\cdots,u_L)\in U^L}\Bigg|\sum_{\ell=1}^{L}\epsilon_{\ell}\varphi_{\ell}(u_{\ell})\Bigg|\right) \leq \mathbb{E}F\left(\sup_{(u_1,\cdots,u_L)\in U^L}\Bigg|\sum_{\ell=1}^{L}\epsilon_{\ell}u_{\ell}\Bigg|\right),
	\]
	where the expected value is taken with respect to the $\epsilon_{\ell}$.
\end{lemma}
In our proof, the contraction inequality will help us deal with the logarithmic terms introduced by the log-likelihood of the Poisson distribution.

We now consider the atomic $M$-norm for tensors \cite{navid1,navid2,navid3}, an approach that will allow our optimal sampling complexity dependence in terms of the tensor dimensions $\{I_n\}_{n=1}^{N}$. First, define
\[
\mathcal{T}_{\pm} \coloneqq \Big\{\boldsymbol{\mathscr{T}}\in\{-1,1\}^{I_1\times\cdots\times I_N} \ | \ \mbox{rank}(\boldsymbol{\mathscr{T}}) = 1 \Big\}.
\]
The atomic $M$-norm of a tensor $\boldsymbol{\mathscr{T}}\in\mathbb{R}^{I_1\times\cdots\times I_N}$ is defined as the gauge (see \cite{Mnorm,Mnorm2}) of $\mathcal{T}_{\pm}$, i.e.,
\[
\|\boldsymbol{\mathscr{T}}\|_M \coloneqq \inf\{t>0 \ | \ \boldsymbol{\mathscr{T}}\in t\ \mbox{conv}(\mathcal{T}_{\pm})\},
\]
where conv$(\mathcal{T}_{\pm})$ is the convex envelope of $\mathcal{T}_{\pm}$. The $M$-norm is a convex norm \cite{navid1,navid2,Mnorm} and we will require the following bounds when acting on bounded rank-$R$ tensors.

\begin{lemma}
	\label{lemma3}
	Assume $\boldsymbol{\mathscr{T}}\in\mathbb{R}^{I_1\times\cdots\times I_N}$ is a rank-$R$ tensor with $\|\boldsymbol{\mathscr{T}}\|_{\infty} \leq \alpha$. Then
	\begin{equation}
		\label{r}
		\|\boldsymbol{\mathscr{T}}\|_M \leq \alpha \left(R\sqrt{R}\right)^{N-1}.
	\end{equation}
	Furthermore, if $\boldsymbol{\mathscr{T}}\in\mathbb{R}_+^{I_1\times\cdots\times I_N}$ with rank$_+(\boldsymbol{\mathscr{T}}) \leq R_+$ then
	\begin{equation}
		\label{r+}
		\|\boldsymbol{\mathscr{T}}\|_M \leq \alpha R_+.
	\end{equation}
\end{lemma}

This result is essentially Theorem 7 in \cite{navid1}, where (\ref{r}) is established. The bound (\ref{r+}) is a simple corollary, which we prove briefly before continuing.
\begin{proof}[Proof of Lemma \ref{lemma3}]
	As discussed, we only need to show (\ref{r+}) using (\ref{r}). By assumption
	\[
	\boldsymbol{\mathscr{T}} = \sum_{r=1}^{R_+}\textbf{a}_{r}^{(1)}\circ \cdots \circ \textbf{a}_{r}^{(N)} \coloneqq \sum_{r=1}^{R_+}\boldsymbol{\mathscr{T}}_r,
	\]
	where each rank one component $\boldsymbol{\mathscr{T}}_r$ in nonnegative. Since $\|\boldsymbol{\mathscr{T}}\|_{\infty}\leq \alpha$, by nonnegativity it is easy to see that $\|\boldsymbol{\mathscr{T}}_r\|_{\infty}\leq \alpha$ for all $r\in [R_+]$. Due to the fact that the $M$-norm is a norm \cite{navid1,navid2,Mnorm}, the triangle inequality gives
	\[
	\|\boldsymbol{\mathscr{T}}\|_M \leq \sum_{r=1}^{R_+}\|\boldsymbol{\mathscr{T}}_r\|_M \leq \sum_{r=1}^{R_+}\alpha = \alpha R_+
	\]
	where the second inequality holds by (\ref{r}) since each $\boldsymbol{\mathscr{T}}_r$ is rank one with $\|\boldsymbol{\mathscr{T}}_r\|_{\infty}\leq \alpha$.
\end{proof}

We also consider the $M$-norm's dual norm
\[
\|\boldsymbol{\mathscr{T}}\|_M^* \coloneqq \max_{\|\boldsymbol{\mathscr{U}}\|_M\leq 1}\langle \boldsymbol{\mathscr{T}},\boldsymbol{\mathscr{U}}\rangle = \max_{\boldsymbol{\mathscr{U}}\in \mathcal{T}_{\pm}}\langle \boldsymbol{\mathscr{T}},\boldsymbol{\mathscr{U}}\rangle,
\]
where the second equality is established in \cite{navid1}. We will require a bound on the expectation of this dual norm when acting on random tensors of a certain structure.

\begin{lemma}
	\label{lemma4}
	Assume $\boldsymbol{\mathscr{V}}\in [-1,1]^{I_1\times \cdots\times I_N}$ is a random tensor with $p$ non-zero entries, which are independent mean zero discrete random variables. Define $\bar{I}\coloneqq \sum_nI_n$ and $\tilde{I}\coloneqq I_1I_2\cdots I_N$. Then, for any $h>0$ such that $\bar{I}-1\geq h\log_2\left(\frac{\tilde{I}}{4\bar{I}}\right)$ we have
	\[
	\mathbb{E}\left(\|\boldsymbol{\mathscr{V}}\|_{M}^{*}\right)^h \leq 2\left(2\sqrt{p\bar{I}}\right)^h.
	\]
\end{lemma}
We postpone the proof of Lemma \ref{lemma4} until Section \ref{proofreqlem}. Lemmas \ref{lemma3} and \ref{lemma4} produce our sampling complexity in terms of $I$ and $R$, where $I = \max_{n\in[N]}I_n$. In contrast to previous approaches that try to generalize results for matrix norms, considering the $M$-norm reduces our sampling complexity from $\mathcal{O}(I^{N/2}\sqrt{R}\log^{3/2}(I))$ \cite{nuctensor} to $\mathcal{O}(I(R\sqrt{R})^{2N-2}\log(I))$ in the general case and $\mathcal{O}(IR^2\log(I))$ in the nonnegative case. Since $R\leq I_1\cdots I_N/I$, this results in a great improvement in many cases. However, the results are still sub-optimal in terms of its rank dependence which is an open problem conjectured to be linear $\mathcal{O}(IR\log(I))$.

\subsubsection{Proof of Lemma \ref{mainlemma0}}
\label{secmainlem0}

We may now proceed to the proof of the main lemma for the zero-truncated case.

\begin{proof}[Proof of Lemma \ref{mainlemma0}]
	We first show (\ref{concentration10}). Afterward, establishing bound (\ref{concentration20}) will only require a slight modification. In what follows, recall that $\Omega$ is fixed and let ${\boldsymbol{\mathscr{U}}}$ be the random tensor with entries $u_{\textbf{i}}$ defined as in the proof of Theorem \ref{mainthm0}. Then, for any $\boldsymbol{\mathscr{T}}\in\mathbb{R}_{+}^{I_1\times\cdots\times I_N}$ we can write
	\[
	\tilde{f}_{\Gamma}(\boldsymbol{\mathscr{T}}) = \sum_{\textbf{i}\in\Omega}x_{\textbf{i}}\log(t_{\textbf{i}})-u_{\textbf{i}}\log(\exp\left(t_{\textbf{i}}\right)-1),
	\]
	which is a sum of independent random variables. We begin by bounding 
	\[
	\mathbb{E}\sup_{\boldsymbol{\mathscr{T}}\in S_R^+(\beta,\alpha)}|\tilde{f}_{\Gamma}(\boldsymbol{\mathscr{T}})-\mathbb{E}\tilde{f}_{\Gamma}(\boldsymbol{\mathscr{T}})|^h
	\]
	for arbitrary $h\geq 1$. Afterward, we will apply Markov's inequality for a specified value of $h$ to obtain the statement with the prescribed probability. To this end, we symmetrize (Lemma \ref{lemma2a}) by introducing a tensor $\boldsymbol{\mathscr{V}}\in\{-1,1\}^{I_1\times\cdots\times I_N}$ whose entries are i.i.d. Rademacher random variables to obtain
	\begin{align*}
		&\mathbb{E}\sup_{\boldsymbol{\mathscr{T}}\in S_R^+(\beta,\alpha)}|\tilde{f}_{\Gamma}(\boldsymbol{\mathscr{T}})-\mathbb{E}\tilde{f}_{\Gamma}(\boldsymbol{\mathscr{T}})|^h \\
		&\leq 2^h\mathbb{E}\sup_{\boldsymbol{\mathscr{T}}\in S_R^+(\beta,\alpha)}\Bigg|\sum_{\textbf{i}\in\Omega}v_{\textbf{i}}\Big[x_{\textbf{i}}\log(t_{\textbf{i}})-u_{\textbf{i}}\log(\exp\left(t_{\textbf{i}}\right)-1)\Big]\Bigg|^h \\
		&\leq 2^{2h-1}\mathbb{E}\sup_{\boldsymbol{\mathscr{T}}\in S_R^+(\beta,\alpha)}\Bigg|\sum_{\textbf{i}\in\Omega}v_{\textbf{i}}x_{\textbf{i}}\log(t_{\textbf{i}})\Bigg|^h
		+ 2^{2h-1}\mathbb{E}\sup_{\boldsymbol{\mathscr{T}}\in S_R^+(\beta,\alpha)}\Bigg|\sum_{\textbf{i}\in\Omega}v_{\textbf{i}}u_{\textbf{i}}\log(\exp\left(t_{\textbf{i}}\right)-1)\Bigg|^h
	\end{align*}
	where the expectations are now over the draw of $\boldsymbol{\mathscr{X}}$ and $\boldsymbol{\mathscr{V}}$ and the last inequality holds since $(a+b)^h \leq 2^{h-1}(a^h+b^h)$ when $a,b>0$ and $h\geq 1$. Both terms resulting from the last inequality can be bounded by applying Lemma \ref{lemma2}. For the first term, define $\varphi(t) \coloneqq \beta\log(t+1)$, which is a contraction for $t\geq \beta - 1$ that vanishes at the origin (see \cite{PoissonMC}). We see that
	\begin{align*}
		2^{2h-1}\mathbb{E}\sup_{\boldsymbol{\mathscr{T}}\in S_R^+(\beta,\alpha)}\Bigg|\sum_{\textbf{i}\in\Omega}\log(t_{\textbf{i}})x_{\textbf{i}}v_{\textbf{i}}\Bigg|^h &= \frac{1}{2}\left(\frac{4}{\beta}\right)^h\mathbb{E}\sup_{\boldsymbol{\mathscr{T}}\in S_R^+(\beta,\alpha)}\Bigg|\sum_{\textbf{i}\in\Omega}\varphi\left(t_{\textbf{i}}-1\right)x_{\textbf{i}}v_{\textbf{i}}\Bigg|^h \\
		&\leq \frac{1}{2}\left(\frac{4}{\beta}\right)^h\mathbb{E}\Big[\max_{\textbf{i}\in\Omega}x_{\textbf{i}}^h\Big]\mathbb{E}\sup_{\boldsymbol{\mathscr{T}}\in S_R^+(\beta,\alpha)}\Bigg|\sum_{\textbf{i}\in\Omega}\varphi(t_{\textbf{i}}-1)v_{\textbf{i}}\Bigg|^h\\
		&\leq \frac{1}{2}\left(\frac{8}{\beta}\right)^h\mathbb{E}\Big[\max_{\textbf{i}\in\Omega}x_{\textbf{i}}^h\Big]\mathbb{E}\sup_{\boldsymbol{\mathscr{T}}\in S_R^+(\beta,\alpha)}\Bigg|\sum_{\textbf{i}\in\Omega}\left(t_{\textbf{i}}-1\right)v_{\textbf{i}}\Bigg|^h,
	\end{align*}
	where the last inequality holds by Lemma \ref{lemma2} since with $\boldsymbol{\mathscr{T}}\in S_R^+(\beta,\alpha)$ we have $t_{\textbf{i}}-1 \geq \beta - 1$ for all $\textbf{i}\in\Omega$. We now bound the two expectations in the last line.
	
	For the term $\mathbb{E}\Big[\max_{\textbf{i}\in\Omega}x_{\textbf{i}}^h\Big]$, we argue as in \cite{PoissonMC} in the proof of Lemma 4. An analogous version of equation (65) therein gives in our context 
	\begin{equation}
		\label{MaxPoisson}
		\mathbb{E}\Bigg[\max_{\textbf{i}\in\Omega}x_{\textbf{i}}^h\Bigg] \leq 2^{2h-1}\left(\alpha^h + \alpha^h(e^2-3)^h + 2h! + \log^{h}(|\Omega|) \right).
	\end{equation}
	
	For the remaining term, let $\Delta_{\Omega}\in\{0,1\}^{I_1\times\cdots\times I_N}$ be the indicator tensor for $\Omega$ and $\textbf{1}\in\{1\}^{I_1\times\cdots\times I_N}$ be the all ones tensor so that
	\begin{align*}
	&\mathbb{E}\sup_{\boldsymbol{\mathscr{T}}\in S_R^+(\beta,\alpha)}\Bigg|\sum_{\textbf{i}\in\Omega}\left(t_{\textbf{i}}-1\right)v_{\textbf{i}}\Bigg|^h \\
	&= \mathbb{E}\sup_{\boldsymbol{\mathscr{T}}\in S_R^+(\beta,\alpha)}|\langle \boldsymbol{\mathscr{T}}-\textbf{1},\boldsymbol{\mathscr{V}}\circ \Delta_{\Omega}\rangle|^h \leq \sup_{\boldsymbol{\mathscr{T}}\in S_R^+(\beta,\alpha)}\|\boldsymbol{\mathscr{T}}-\textbf{1}\|_M^h\mathbb{E}\left(\|\boldsymbol{\mathscr{V}}\circ \Delta_{\Omega}\|_M^*\right)^h,
	\end{align*}
	where the inequality holds by the definition of the dual norm. Applying equation (\ref{r+}) from Lemma \ref{lemma3} and the fact that the $M$-norm is a norm \cite{navid1,navid2,Mnorm}, we have by the triangle inequality
	\begin{equation}
		\label{eq2}
		\|\boldsymbol{\mathscr{T}}-\textbf{1}\|_M \leq \|\boldsymbol{\mathscr{T}}\|_M + \|\textbf{1}\|_M \leq R_M^++1,
	\end{equation}
	where the second inequality holds since $\boldsymbol{\mathscr{T}}\in S_R^+(\beta,\alpha)$, rank$_+(\textbf{1})=1$ with $\|\textbf{1}\|_{\infty} = 1$, and by definition of $R_M^+$ (max $M$-norm over $S_{R}^{+}(\alpha,\beta)$). Furthermore, $\boldsymbol{\mathscr{V}}\circ \Delta_{\Omega}$ satisfies the conditions of Lemma \ref{lemma4}, so assuming $h$ will be chosen such that 
	\begin{equation}
		\label{hcond}
		\sum_{n=1}^{N}I_n\geq h\log_2\left(\frac{I_1\cdots I_N}{4\sum_{n=1}^{N}I_n}\right) + 1
	\end{equation}
	we have
	\[
	\mathbb{E}\left(\|\boldsymbol{\mathscr{V}}\circ \Delta_{\Omega}\|_M^*\right)^h \leq 2\left(2\sqrt{|\Omega|\sum_{n=1}^{N}I_n}\right)^h.
	\]
	Thus far, we have shown
	\begin{align*}
		&2^{2h-1}\mathbb{E}\sup_{\boldsymbol{\mathscr{T}}\in S_R^+(\beta,\alpha)}\Bigg|\sum_{\textbf{i}\in\Omega}v_{\textbf{i}}x_{\textbf{i}}\log(t_{\textbf{i}})\Bigg|^h \\
		&\leq \frac{1}{2}\left(\alpha^h + \alpha^h(e^2-3)^h + 2h! + \log^{h}(|\Omega|) \right)\left(\frac{64(R_M^++1)}{\beta}\sqrt{|\Omega|\sum_{n=1}^{N}I_n}\right)^h.
	\end{align*}
	
	The remaining term can be bounded in a similar manner, considering $\phi(t) := (1-e^{-\beta})\log(\exp\left(t + \log(2)\right)-1)$ which is a contraction for $t\geq \beta - \log(2)$ that vanishes at the origin. Using Lemma \ref{lemma2} again we obtain
	\begin{align*}
		&2^{2h-1}\mathbb{E}\sup_{\boldsymbol{\mathscr{T}}\in S_R^+(\beta,\alpha)}\Bigg|\sum_{\textbf{i}\in\Omega}v_{\textbf{i}}u_{\textbf{i}}\log(\exp\left(t_{\textbf{i}}\right)-1)\Bigg|^h \\
		&= \frac{2^{2h-1}}{(1-e^{-\beta})^h}\mathbb{E}\sup_{\boldsymbol{\mathscr{T}}\in S_R^+(\beta,\alpha)}\Bigg|\sum_{\textbf{i}\in\Omega}v_{\textbf{i}}u_{\textbf{i}}\phi(t_{\textbf{i}}-\log(2))\Bigg|^h\\
		&\leq \frac{2^{3h-1}}{(1-e^{-\beta})^h}\mathbb{E}\sup_{\boldsymbol{\mathscr{T}}\in S_R^+(\beta,\alpha)}\Bigg|\sum_{\textbf{i}\in\Omega}v_{\textbf{i}}u_{\textbf{i}}(t_{\textbf{i}}-\log(2))\Bigg|^h \leq \left(\frac{16(R_M^++1)}{1-e^{-\beta}}\sqrt{|\Omega|\sum_{n=1}^{N}I_n}\right)^h,
	\end{align*}
	where the last inequality holds as in the bound of the first term by considering the $M$-norm, its dual, and applying Lemma \ref{lemma3} to $\boldsymbol{\mathscr{T}}-\log(2)$ and Lemma \ref{lemma4} to $\boldsymbol{\mathscr{U}}\circ\boldsymbol{\mathscr{V}}\circ\Delta_{\Omega}$.
	
	In conclusion, we have shown
	\begin{align*}
		\mathbb{E}\sup_{\boldsymbol{\mathscr{T}}\in S_R^+(\beta,\alpha)}|\tilde{f}_{\Gamma}(\boldsymbol{\mathscr{T}})-\mathbb{E}\tilde{f}_{\Gamma}(\boldsymbol{\mathscr{T}})|^h \leq \delta_0 , 
	\end{align*}
	where
	\begin{align*}
		\delta_0 \coloneqq \, &\frac{1}{2}\left(\alpha^h + \alpha^h(e^2-3)^h + 2h! + \log^{h}(|\Omega|) \right)\left(\frac{64(R_M^++1)}{\beta}\sqrt{|\Omega|\sum_{n=1}^{N}I_n}\right)^h \\
		&+ \left(\frac{16(R_M^++1)}{1-e^{-\beta}}\sqrt{|\Omega|\sum_{n=1}^{N}I_n}\right)^h.
	\end{align*}
	Applying Markov's inequality, we have for any $\delta>0$
	\begin{align*}
		\mathbb{P}\left(\sup_{\boldsymbol{\mathscr{T}}\in S_R^+(\beta,\alpha)}|\tilde{f}_{\Gamma}(\boldsymbol{\mathscr{T}})-\mathbb{E}\tilde{f}_{\Gamma}(\boldsymbol{\mathscr{T}})|\geq \delta\right) &= \mathbb{P}\left(\sup_{\boldsymbol{\mathscr{T}}\in S_R^+(\beta,\alpha)}|\tilde{f}_{\Gamma}(\boldsymbol{\mathscr{T}})-\mathbb{E}\tilde{f}_{\Gamma}(\boldsymbol{\mathscr{T}})|^h\geq \delta^h\right)\\ 
		&\leq \frac{\mathbb{E}\sup_{\boldsymbol{\mathscr{T}}\in S_R^+(\beta,\alpha)}|\tilde{f}_{\Gamma}(\boldsymbol{\mathscr{T}})-\mathbb{E}\tilde{f}_{\Gamma}(\boldsymbol{\mathscr{T}})|^h}{\delta^h} \leq \frac{\delta_0}{\delta^h}.
	\end{align*}
	Pick $\delta = 2\delta_0^{1/h}$ and $h=\log_2(|\Omega|)$, so that
	\[
	\mathbb{P}\left(\sup_{\boldsymbol{\mathscr{T}}\in S_R^+(\beta,\alpha)}|\tilde{f}_{\Gamma}(\boldsymbol{\mathscr{T}})-\mathbb{E}\tilde{f}_{\Gamma}(\boldsymbol{\mathscr{T}})|\geq 2\delta_0^{1/h}\right) \leq 2^{-h} = \frac{1}{|\Omega|}.
	\]
	Using $(a^h+b^h)^{1/h} \leq a+b$, $h!^{1/h} \leq h$, and $(a^h+b^h+c^h+d^h)^{1/h} \leq a+b+c+d$ if $a,b,c,d>0$, we can simplify the bound as
	\begin{align*}
		2\delta_0^{1/h} &\leq \left(\alpha(e^2-2) + 3\log_2(|\Omega|) \right)\frac{128(R_M^++1)}{\beta}\sqrt{|\Omega|\sum_{n=1}^{N}I_n} + \frac{32(R_M^++1)}{1-e^{-\beta}}\sqrt{|\Omega|\sum_{n=1}^{N}I_n}\\
  &=32(R_M^++1)\left(\alpha(e^2-2) + 3\log_2(|\Omega|) \right)\left(\frac{4}{\beta} + \frac{\left(\alpha(e^2-2) + 3\log_2(|\Omega|) \right)^{-1}}{\left(1-e^{-\beta}\right)}\right)\sqrt{|\Omega|\sum_{n=1}^{N}I_n} \\
		&= 32(R_M^++1)\left(\frac{(4+\beta\tau)e^{\beta}-4}{(e^{\beta}-1)\beta\tau}\right)\sqrt{|\Omega|\sum_{n=1}^{N}I_n},
	\end{align*}
 where in the last equality we define $\tau^{-1} \coloneqq \alpha(e^2-2) + 3\log_2(|\Omega|)$. We note that (\ref{hcond}) with our choice $h=\log_2(|\Omega|)$ is satisfied if
	\[
	\min_nI_n \geq (N-1)\log_2^2\left(\max_{n}I_n\right) + \frac{1}{N}.
	\]
	which holds under our assumed contexts defined in Theorems~\ref{simplethm}, \ref{mainthm0}, and \ref{mainthm}.
	
	To obtain (\ref{concentration20}), we use an analogous argument with respect to $S_R(\beta,\alpha)$ which replaces the term $R_M^+$ with $R_M$ and otherwise leaves all other terms unchanged, thereby establishing (\ref{concentration20}) with the same probability.
\end{proof}

\subsubsection{Proof of Lemma \ref{mainlemma}}

Here we prove the main lemma of the Poisson tensor completion result. The proof is very similar to strategy used in the last section and for brevity we will apply bounds therein.

\begin{proof}[Proof of Lemma \ref{mainlemma}]
	We first show (\ref{concentration1}). Afterward, establishing bound (\ref{concentration2}) will only require a slight modification. Notice that
	\[
	f_{\Omega}(\boldsymbol{\mathscr{T}})-\mathbb{E}f_{\Omega}(\boldsymbol{\mathscr{T}}) = \sum_{\textbf{i}\in\Omega}\log(t_{\textbf{i}})\left(x_{\textbf{i}}-\mathbb{E}x_{\textbf{i}}\right),
	\]
	where, with $\Omega$ fixed, we take expected value with respect to $\boldsymbol{\mathscr{X}}$. To bound 
	\[
	\mathbb{E}\sup_{\boldsymbol{\mathscr{T}}\in S_R^+(\beta,\alpha)}|f_{\Omega}(\boldsymbol{\mathscr{T}})-\mathbb{E}f_{\Omega}(\boldsymbol{\mathscr{T}})|^h
	\]
	for arbitrary $h\geq 1$, we apply Lemma \ref{lemma2a} so that
	\[
	\mathbb{E}\sup_{\boldsymbol{\mathscr{T}}\in S_R^+(\beta,\alpha)}|f_{\Omega}(\boldsymbol{\mathscr{T}})-\mathbb{E}f_{\Omega}(\boldsymbol{\mathscr{T}})|^h \leq 2^h\mathbb{E}\sup_{\boldsymbol{\mathscr{T}}\in S_R^+(\beta,\alpha)}\Bigg|\sum_{\textbf{i}\in\Omega}\log(t_{\textbf{i}})x_{\textbf{i}}v_{\textbf{i}}\Bigg|^h,
	\]
	where $\boldsymbol{\mathscr{V}}\in\{-1,1\}^{I_1\times\cdots\times I_N}$ is a random tensor whose entries are i.i.d. Rademacher random variables and the expectation is now over the draw of $\boldsymbol{\mathscr{X}}$ and $\boldsymbol{\mathscr{V}}$. This last term can be bounded exactly as in the proof of Lemma \ref{mainlemma0}, to obtain
	\begin{align*}
		&\mathbb{E}\sup_{\boldsymbol{\mathscr{T}}\in S_R^+(\beta,\alpha)}|f_{\Omega}(\boldsymbol{\mathscr{T}})-\mathbb{E}f_{\Omega}(\boldsymbol{\mathscr{T}})|^h \leq \delta_0 ,
	\end{align*}
	where
	\begin{align*}
		\delta_0 \coloneqq \left(\alpha^h + \alpha^h(e^2-3)^h + 2h! + \log^{h}(|\Omega|) \right)\left(\frac{32(R_{M}^{+}+1)}{\beta}\sqrt{|\Omega|\sum_{n=1}^{N}I_n}\right)^h.
	\end{align*}
	Applying Markov's inequality, we have for any $\delta>0$
	\begin{align*}
		\mathbb{P}\left(\sup_{\boldsymbol{\mathscr{T}}\in S_R^+(\beta,\alpha)}|f_{\Omega}(\boldsymbol{\mathscr{T}})-\mathbb{E}f_{\Omega}(\boldsymbol{\mathscr{T}})|\geq \delta\right) 
		&= \mathbb{P}\left(\sup_{\boldsymbol{\mathscr{T}}\in S_R^+(\beta,\alpha)}|f_{\Omega}(\boldsymbol{\mathscr{T}})-\mathbb{E}f_{\Omega}(\boldsymbol{\mathscr{T}})|^h\geq \delta^h\right)\\ 
		&\leq \frac{\mathbb{E}\sup_{\boldsymbol{\mathscr{T}}\in S_R^+(\beta,\alpha)}|f_{\Omega}(\boldsymbol{\mathscr{T}})-\mathbb{E}f_{\Omega}(\boldsymbol{\mathscr{T}})|^h}{\delta^h} \leq \frac{\delta_0}{\delta^h}.
	\end{align*}
	Pick $\delta = 2\delta_0^{1/h}$ and $h=\log_2(|\Omega|)$, so that
	\[
	\mathbb{P}\left(\sup_{\boldsymbol{\mathscr{T}}\in S_R^+(\beta,\alpha)}|f_{\Omega}(\boldsymbol{\mathscr{T}})-\mathbb{E}f_{\Omega}(\boldsymbol{\mathscr{T}})|\geq 2\delta_0^{1/h}\right) \leq 2^{-h} = \frac{1}{|\Omega|}.
	\]
	For the advertised result, we further bound
	\[
	\left(\alpha^h + \alpha^h(e^2-3)^h + 2h! + \log^{h}(|\Omega|) \right)^{1/h} \leq \alpha(e^2-2) + 3\log_2(|\Omega|).
	\]
	
	To obtain (\ref{concentration2}), we use an analogous argument with respect to $S_R(\beta,\alpha)$ and $R_M$.
\end{proof}

\subsection{Proof of Additional Lemmas}
\label{proofreqlem}

From Section \ref{reqlem}, we need only to prove Lemma \ref{lemma4} since the remaining lemmas are established in the respective citations. To obtain the lemma, we will use the following result for bounded discrete random variables.
\begin{theorem}
	\label{thmthm}
	Let $y\in[0,L]$ be a discrete random variable. If for some $\delta\in(0,\infty)$ we have
	\[
	\mathbb{P}\left(y\geq \delta\right) \leq \frac{\delta}{L},
	\]
	then
	\[
	\mathbb{E}y\leq 2\delta.
	\]
\end{theorem}
The proof of Theorem \ref{thmthm} is rather simple, we quickly provide the proof before continuing.

\begin{proof}[Proof of Theorem \ref{thmthm}]
	If $\delta \geq L$, then the conclusion is trivial. Otherwise, let $y_1 < y_2 < y_3 < \cdots \leq L$ be the possible outcomes of $y$ and let $k_0\in \mathbb{N}$ be such that $y_{k_0}\leq \delta < y_{k_0+1}$. Then
	\begin{align*}
		\mathbb{E}y &= \sum_{k=1}^{\infty}y_k\mathbb{P}\left(y=y_k\right) = \sum_{k=1}^{k_0}y_k\mathbb{P}\left(y=y_k\right) + \sum_{k=k_0+1}^{\infty}y_k\mathbb{P}\left(y=y_k\right)\\
		&\leq \delta\sum_{k=1}^{k_0}\mathbb{P}\left(y=y_k\right) + L\sum_{k=k_0+1}^{\infty}\mathbb{P}\left(y=y_k\right) = \delta\mathbb{P}\left(y\leq y_{k_0}\right) + L\mathbb{P}\left(y\geq y_{k_0+1}\right) \\
		&\leq \delta + L\frac{\delta}{L} = 2\delta.
	\end{align*}
\end{proof}

With this in mind, we now proceed to the proof of Lemma \ref{lemma4}.

\begin{proof}[Proof of Lemma \ref{lemma4}]
	Recall that we have defined $\bar{I}\coloneqq \sum_{n=1}^{N}I_n$, $\tilde{I}\coloneqq I_1I_2\cdots I_N$ and let $\Omega\subset [I_1]\times\cdots\times [I_N]$ be the set of $p$ non-zero entries of $\boldsymbol{\mathscr{V}}$. Using equation (4.41) in \cite{navid2} we have
	\[
	\|\boldsymbol{\mathscr{V}}\|_M^* \leq \sup_{\boldsymbol{\mathscr{U}}\in \mathcal{T}_{\pm}}\Bigg|\sum_{\textbf{i}\in\Omega}v_{\textbf{i}}u_{\textbf{i}}\Bigg|.
	\]
	Notice that the term on the right hand side is a discrete random variable, taking values in $[0,p]$. For fixed $\boldsymbol{\mathscr{U}}\in \mathcal{T}_{\pm}$, a standard Hoeffding's inequality for bounded random variables gives for $t>0$
	\[
	\mathbb{P}\left(\Bigg|\sum_{\textbf{i}\in\Omega}v_{\textbf{i}}u_{\textbf{i}}\Bigg|\geq t\right) \leq 2\exp\left(-\frac{t^2}{2p}\right).
	\]
	Since $|\mathcal{T}_{\pm}|\leq 2^{\bar{I}}$ (see \cite{navid1,navid2}), a union bound and choosing $t=2\sqrt{p\bar{I}}$ provides
	\[
	\mathbb{P}\left(\sup_{\boldsymbol{\mathscr{U}}\in \mathcal{T}_{\pm}}\Bigg|\sum_{\textbf{i}\in\Omega}v_{\textbf{i}}u_{\textbf{i}}\Bigg|\geq 2\sqrt{p\bar{I}}\right) \leq 2^{\bar{I}+1}e^{-2\bar{I}} \leq e^{\bar{I}+1}e^{-2\bar{I}} = e^{-\bar{I}+1}.
	\]
	Equally, for any $h>0$ we have shown
	\[
	\mathbb{P}\left(\sup_{\boldsymbol{\mathscr{U}}\in \mathcal{T}_{\pm}}\Bigg|\sum_{\textbf{i}\in\Omega}v_{\textbf{i}}u_{\textbf{i}}\Bigg|^h\geq \left(2\sqrt{p\bar{I}}\right)^h\right) \leq e^{-\bar{I}+1},
	\]
	and by Theorem \ref{thmthm} we end the proof if $e^{-\bar{I}+1} < \frac{\left(2\sqrt{p\bar{I}}\right)^h}{p^h}$. To this end, using the Maclaurin series of the exponential function, we note that for any $\ell\in\mathbb{N}$
	\[
	\exp(2(\bar{I}-1)/h)\geq \frac{2^{\ell}(\bar{I}-1)^{\ell}}{h^{\ell}\ell !}\geq \frac{2^{\ell}(\bar{I}-1)^{\ell}}{(h\ell)^{\ell}}
	\]
	which in particular holds for the non-integer choice $\ell\coloneqq \log_2(\tilde{I}/(4\bar{I}))$ in the last term. Recall our assumption $\bar{I}-1 \geq h\log_2(\tilde{I}/(4\bar{I})) = h\ell$, so that
	\[
	p \leq \tilde{I} = 2^{\ell}4\bar{I} \leq \frac{2^{\ell}4\bar{I}(\bar{I}-1)^{\ell}}{(h\ell)^{\ell}} \leq 4\bar{I}\exp\left(2(\bar{I}-1)/h\right).
	\]
	We have shown $p \leq 4\bar{I}\exp\left(2(\bar{I}-1)/h\right)$, raising both sides to the power of $h/2$ and rearranging gives the desired inequality. We conclude $\mathbb{E}\left(\|\boldsymbol{\mathscr{V}}\|_{M}^{*}\right)^h\leq 2\left(2\sqrt{p\bar{I}}\right)^h$.
	
\end{proof}

%% file: impldetails.tex
\definecolor{dkgreen}{rgb}{0,0.6,0}
\definecolor{gray}{rgb}{0.5,0.5,0.5}
\definecolor{mauve}{rgb}{0.58,0,0.82}
\lstset{
  language=Matlab,                
  basicstyle=\scriptsize,           
  backgroundcolor=\color{white},      
  showspaces=false,               
  showstringspaces=false,         
  showtabs=false,                 
  frame=single,                   
  rulecolor=\color{black},        
  tabsize=2,                      
  captionpos=b,                   
  breaklines=true,                
  breakatwhitespace=false,        
  title=\lstname,                   
  keywordstyle=\color{blue},          
  commentstyle=\color{dkgreen},       
  stringstyle=\color{mauve},         
  escapeinside={\%*}{*)},            
  morekeywords={*,...}               
}

Experiments were conducted using Tensor Toolbox for MATLAB v3.2.1~\cite{TTB_src} in MATLAB R2022b. The MATLAB function \texttt{poissrand} from MATLAB's Statistics and Machine Learning Toolbox v12.2 was also used in the experiments.

\begin{center}
\begin{minipage}{0.925\hsize}%
\begin{lstlisting}[frame=single,framesep=12pt,linewidth=0.98\textwidth,caption={Helper MATLAB function for generating a parameter tensor $\boldsymbol{\mathscr{M}}$.}]
function M = tensor_ztp_create_param_tensor(dim,R,param_range)
    N = length(dim);
    factor_range = param_range.^(1/N)/R^(1/N);   
    % Call create_problem from the Tensor Toolbox for MATLAB
    M = create_problem('Size', dim, 'Num_Factors', R, ...
        'Factor_Generator', @(m,n)(factor_range(1)+...
            (rand(m,n))*(factor_range(2)-factor_range(1))), ...
        'Lambda_Generator', @(m,n)ones(m,1), 'Noise', 0);
    M = normalize(arrange(M.Soln));
\end{lstlisting}
\end{minipage}
\end{center}

\begin{center}
\begin{minipage}{0.925\hsize}%
\begin{lstlisting}[frame=single,framesep=12pt,linewidth=0.98\textwidth,caption={Helper MATLAB function for running an individual experiment.}]
function [E_poisson,E_oracle,E_ztp] = tensor_ztp_run_experiment(...
    N, I, R, p, reps, nstarts, b, a, reg_val, opts_gcp, filename)

% initialize matrices to store relative errors
E_poisson = zeros(reps,length(p),nstarts); E_oracle = E_poisson; E_ztp = E_poisson;

% Poisson NLL function/gradient using reg_val for regularization
f_poisson = @(x,m) m - x.*log(m + reg_val);
g_poisson = @(x,m) 1 - x./(m + reg_val);
% Poisson NLL function/gradient using reg_val for regularization
f_ztp = @(x,m) f_poisson(x,m) + log(1 - exp(-m) + reg_val);
g_ztp = @(x,m) g_poisson(x,m) + 1./((exp(m) - 1) + reg_val);

% Generate low-rank random tensor with entries in [b, a]
M = tensor_ztp_create_param_tensor(I*ones(1,N), R, [b a]);

% Main loop
for k1 = 1:reps
    % Generate Poisson observations
    rng(k1); X_obs = poissrnd(double(full(M))); % Poisson observations
    for k2 = 1:length(p)
        % Generate missing entries with desired percentage
        rng(k2); OmC = randperm(I^N);  % random indices, Omega^C
        OmC = OmC(1:round(p(k2)*I^N)); % indices of unobseved entries
        X = X_obs;                     % copy from X_obs for each k2
        X(OmC) = 0;                    % inject false zeros into X
        X = tensor(X);                 % tensor version of data
        for k3 = 1:nstarts
            % Poisson parameter estimation, ALL zeros observed
            rng(k3); Mhat_poisson = gcp_opt(X,R,opts_gcp,'func',f_poisson, ...
                'grad',g_poisson,'lower',0);
            E_poisson(k1,k2,k3) = norm(M-Mhat_poisson)/norm(M);
            % Oracle parameter estimation, only true zeros (Omega is known)
            W2 = ones(I*ones(1,N));  % create an indicator tensor for mask
            W2(OmC) = 0;             % remove false zeros using Omega^C
            W2 = tensor(W2);
            rng(k3); Mhat_oracle = gcp_opt(X,R,opts_gcp,'func',f_poisson, ...
                'grad',g_poisson,'lower',0,'mask',W2); 
            E_oracle(k1,k2,k3) = norm(M-Mhat_oracle)/norm(M);
            % ZTP parameter estimation, ignore ALL zeros
            ind = find(X>0);             % find nonzeros in X
            W = tensor(@zeros,size(X));  % create an indicator tensor for mask
            W(ind) = 1;                  % indicate where nonzeros in X are
            Gam = find(X(:)>0);
            rng(k3); Mtilde_ztp = gcp_opt(X,R,opts_gcp,'func',f_ztp, ...
                'grad',g_ztp,'lower',0,'mask',W); 
            E_ztp(k1,k2,k3) = norm(M-Mtilde_ztp)/norm(M);
        end
    end
end

% Save outputs as .mat file
save(filename,'N','I','R','b','a','p','reg_val','E_poisson','E_oracle','E_ztp');
\end{lstlisting}
\end{minipage}
\end{center}

\begin{center}
\begin{minipage}{0.925\hsize}%
\begin{lstlisting}[frame=single,framesep=12pt,linewidth=0.98\textwidth,caption={Main MATLAB script for reproducing experiments.}]
% Tensor parameters
N = 3;                     % number of dimensions
I_array = [50, 100, 200];  % size per dimension, multiple experiments
R = 5;                     % rank

% Experiment parameters
p = [0:.05:.95 0.96:0.01:0.99];  % percent missing entries
reps = 50;                       % number of runs per experiment
random_seed = 12345;             % for reproducibility

% GCP optimization parameters
clear opts_gcp
opts_gcp.opt = 'lbfgsb';   % Limited-memory bound-constrained quasi-Newton
opts_gcp.maxiters = 3000;  % maximum number of iters
opts_gcp.printitn = 1000;  % number of iterations before printing output
opts_gcp.pgtol = 1e-12;    % stopping tolerance - gradient
opts_gcp.factr = 1e-10;    % stopping tolerance - function value reduction

% Function and gradient resularization
reg_val = 1e-10;

%% Experiments, looping over I_array
for I = I_array
    % Varying $\beta$
    a = 2.5; betas = [1, .1,.01,.001];
    for i = 1:length(betas)
        b = betas(i);
        filename = sprintf('results_I_%d_beta_%f_alpha_%f.mat',I,b,a);
        rng(random_seed); 
        [E_poisson,E_oracle,E_ztp] = tensor_ztp_run_experiment(...
            N,I,R,p,reps,b,a,reg_val,opts_gcp,filename);
    end
    % Varying $\alpha$
    b = 0.1; alphas = [5,10,25,50];
    for i = 1:length(alphas)
        a = alphas(i);
        filename = sprintf('results_I_%d_beta_%f_alpha_%f.mat',I,b,a);
        rng(random_seed); 
        [E_poisson,E_oracle,E_ztp] = tensor_ztp_run_experiment(...
            N,I,R,p,reps,b,a,reg_val,opts_gcp,filename);
    end
end
\end{lstlisting}
\end{minipage}
\end{center}



